\documentclass{llncs}

\usepackage[linesnumbered,vlined,ruled]{algorithm2e}
\usepackage{graphicx, amsmath, amssymb}
\usepackage{times}
  \newtheorem{notation}[theorem]{Notation}

\newcommand{\pw}{\operatorname{pw}}

\newcommand{\sm}{\setminus}
\newcommand{\cO}{{\cal O}}

\newcommand{\tQ}{\tilde{Q}}

\title{\textsc{Treewidth} and \textsc{Pathwidth} parameterized by vertex cover}

\author{
Mathieu Chapelle\inst{1} \and
Mathieu Liedloff\inst{2} \and 
Ioan Todinca\inst{2} \and 
Yngve Villanger\inst{3}
}

\institute{
IGM-LabInfo, UniversitŽ Paris-Est Marne-la-Vall\'ee, 5 Bd Descartes - Champs sur Marne
77454 Marne la Vall\'ee cedex 2,
France, \texttt{mathieu.chapelle@univ-mlv.fr}
\and 
LIFO, Universit{\'e} d'Orl{\'e}ans, BP 6759, F-45067 Orl{\'e}ans Cedex 2, France,\\ \texttt{(mathieu.liedloff $\mid$ ioan.todinca)@univ-orleans.fr}
\and
Department of Informatics, University of Bergen, N-5020 Bergen, Norway,\\ \texttt{yngve.villanger@uib.no}
}

\date{\today}

\begin{document}
\pagestyle{plain}

\maketitle

\begin{abstract}
After the number of vertices, 
{\it Vertex Cover} is the largest of the classical graph parameters and 
has more and more frequently been used as a separate parameter in 
parameterized problems, including problems that are not directly related to the {\it Vertex Cover}.
Here we consider the {\sc treewidth} and {\sc pathwidth} problems parameterized by 
$k$, the size of a minimum vertex cover of the input graph. We show that the 
{\sc pathwidth} and {\sc treewidth} can be computed in $O^*(3^k)$ time. 
This complements 
recent polynomial kernel results for {\sc treewidth} and {\sc pathwidth} 
parameterized by the {\it Vertex Cover}. 
\end{abstract}

\section{Introduction}

Parameterized algorithms are typically used in the setting where the provided 
problem is {\sc NP}-hard and we want to bound the exponential part of the running time 
to a function of some specific parameter. 
This parameter can be any property related to the 
input, the output, or the problem itself. A classical parameter is $n$, the size of the input 
or the number of vertices in the input graph. Algorithms of this type are usually refered to as 
{\it moderately exponential time} algorithms \cite{FominK10}, and in many cases it is non trivial to improve the exponential dependence on $n$ to something better than the 
naive brute force bound. 

The number of vertices is not the only natural graph parameter; there are also parameters like 
{\it treewidth}, 
{\it feedback vertex set}, and 
{\it vertex cover}. For every graph, there is an increasing order on these parameters:  
{\it treewidth} is the smallest, and then 
{\it feedback vertex set}, 
{\it vertex cover} and eventually $n$ come in this order. 
We refer to Bodlaender et. al.~\cite{BodlaenderJK11} for more parameters and the relation between them. 
Many moderately exponential time algorithms have an exponential dependence on $n$ that 
is of the form $c^n$ for some constant $c<2$. 
When the exponential part of the running time is bounded by one of the other graph 
parameters, we typically see a much faster growing function than we do for parameter $n$. 
Thus, we have reached a situation where tradeoffs can be made between the size the parameter we choose and the exponential dependence on this parameter.

We  use a modified big-Oh notation that suppresses all other (polynomially bounded) terms.
Thus for functions $f$ and $g$ we write $f(n,k)=O^*(g(n,k))$ if
$f(n,k)=O(g(n,k)\cdot n^{\cO(1)})$.
Consider the problems of computing the {\sc treewidth} or the  {\sc pathwidth} of a 
 given graph $G$. For parameter $n$ both these values can be computed in $O^*(2^n)$ 
by a  dynamic programming approach proposed by Held and Karp~\cite{HeldK62}. Currently the best moderately exponential time 
algorithms for these problems have $O(1.735^n)$ \cite{FominV10} and $O(1.89^n)$ \cite{KitsunaiKKTT12} running times respectively. 
On the other hand if we go to the smaller parameters {\it treewidth} and 
{\it pathwidth} the best known running times are of the from $O^*(2^{O({k^3})})$~\cite{Bodlaender96}.  
Thus, it is preferable to use the $O^*(2^{O({k^3})})$ algorithm parameterized by {\it treewidth} if the {\it treewidth} is $O(n^{1/3})$, and the algorithms parameterized by $n$ otherwise. 
In this paper we are considering {\it vertex cover} as a parameter for 
the {\sc treewidth} or the {\sc pathwidth} problems. Our objective is then to find the 
most efficient algorithm for these two problems where the exponential part of the 
running time is only depending on the size of the {\it vertex cover}. 

Using the size of the {\it vertex cover} as a parameter when analyzing algorithms and 
solving problems is not a new idea. Some examples from the literature  are 
an  $O^*(2^k)$ algorithm for {\sc cutwidth} parameterized by {\it vertex cover} \cite{CyganLPPS11a},
an $O^*(2^k)$ algorithm for {\sc chordal graph} sandwich parameterized by the {\it vertex cover} of an edge set \cite{HeggernesMNV10}, and 
different variants of graph layout problems parameterized by {\it vertex cover} \cite{FellowsLMRS08}. 

Another direction in the area of parameterized complexity is kernelization or instance 
compression. Recently it was shown \cite{Drucker12} that we can not expect that the {\sc treewidth} and {\sc pathwidth} problems have a polynomial kernel unless $NP \subseteq coNP/poly$ when parameterized by {\it treewidth} or {\it path\-width}, but on the other hand they do have a kernel of size 
$O(k^3)$ when parameterized by {\it vertex cover}~\cite{BodlaenderJK11,BodlaenderJK12}.
Existence of a polynomial size kernel does not necessarily imply the existence of an algorithm that has a slow growing exponential function in the size of the parameter. 
Indeed if we first kernelize then use the best moderately exponential time algorithm of~\cite{FominV10} on the kernel, we still obtain an $O^*(2^{O({k^3})})$ algorithm for
\textsc{treewidth}
parameterized by {\it vertex cover}. 
Hence dependence in the parameter is still similar to the algorithm parameterized by \textit{treewidth}~\cite{Bodlaender96}. 

\paragraph{Our results.}
We provide an $O^*(3^{k})$ time algorithm for {\sc pathwidth} and 
{\sc treewidth} when parameterized by $vc$ the size of the {\it vertex cover}. 
It means that this algorithm will be preferable for graphs where the 
{\it treewidth} is $\Omega(vc^{1/3})$ and the vertex cover is at most $0.5n$ and $0.58n$ 
for the {\sc treewidth} and {\sc pathwidth} problems respectively. 
Another consequence is that the {\sc treewidth} and {\sc pathwidth} of a bipartite graph can be computed in $O^*(3^{n/2})$ or $O^*(1.733^{n})$ time, which is better than the running time provided by the corresponding moderate exponential time algorithms ($O^*(1.735^n)$ \cite{FominV10} and $O^*(1.89^n)$ \cite{KitsunaiKKTT12} respectively). 
We point out that for {\sc treewidth}, we first provide an $O^*(4^{k})$ algorithm based on dynamic programming. The algorithm is then modified to obtain a running time of 
$O^*(3^{k})$, and for this purpose we use the subset convolution technique introduced in~\cite{BHKK07}.

In addition to this we also show (in Appendix~\ref{se:pwcomp}) that the {\sc pathwidth} can be computed in 
$O^*(2^{k'})$ time where $k'$ is the {\it vertex cover} size of the complement of the graph. 
This matches the result of~\cite{BodlaenderFKKT06} for {\sc treewidth} parameterized 
by the {\it vertex cover} size of the complement of the graph. 

%

\section{Preliminaries}

All graphs considered in this article are simple and undirected. For a graph $G=(V,E)$
we denote by $n = |V|$ the number of vertices and by $m = |E|$ the number of edges.
The neighborhood of a vertex $v$ is defined as $N(v) = \{u \in V:\{u,v\} \in E\}$,
and the closed neighborhood is defined as $N[v] = N(v) \cup \{v\}$. 
For a vertex set $W$, we define its neighborhood as $N(W) = \bigcup_{v \in W} N(v) \setminus W$,
and its closed neighborhood as $N[W] = N(W)\cup W$. 
A vertex set $C \subseteq V$ in a graph $G=(V,E)$ is called a 
{\it vertex cover} if for every edge $uv \in E(G)$ we have that 
vertex $u$ or $v$ is in $C$. By extension, the minimum size of a vertex cover of $G$ is usually called the {\it vertex cover} of $G$. 
Vertex set $X$ is called a clique of $G$ if for each pair $u,v \in X$ 
we have that $uv \in E$.

%
%

\begin{proposition}[\cite{ChenKX10}]\label{pr:VC}
The minimum vertex cover problem can be solved in time $O^*(1.28^k)$, where $k$ is the vertex cover of the input graph. 
\end{proposition}


We now define \emph{tree} and \emph{path decompositions}.
A \emph{tree decomposition} of a graph $G=(V,E)$ is a pair $(T, {\cal X})$ where $T=(I,F)$ is a tree and ${\cal X}=\{X_i \mid i \in I\}$ is a family of subsets of $V$, called \emph{bags}, where 
\begin{itemize}
\item $V = \bigcup_{i \in I} X_i$, 
\item for each edge $uv \in E$ there exists an $i \in I$ such that $u,v \in X_i$, and 
\item for each vertex $v \in V$ the nodes $\{i \in I \mid v \in X_i\}$ induce a (connected) subtree of $T$.
\end{itemize}
The \emph{width} of the tree decomposition $(T, {\cal X})$ is $\max_{i \in I} |X_i|-1$ (the maximum size of a bag, minus one) and the \emph{treewidth} of $G$ is the minimum width over all tree decompositions of $G$. 

A \emph{path decomposition} of $G$ is a tree decomposition $(T, {\cal X})$ such that the tree $T$ is actually a path. The \emph{pathwidth} of $G$ is the minimum width over all path decomposition of $G$. 

The following result is a straightforward consequence of Helly's property for a family of subtrees of a tree.
\begin{proposition}\label{pr:hellyc}
Let $(T, {\cal X})$ be a tree decomposition of graph $G=(V,E)$. Let $H=(V,F)$ be the graph such that $xy \in F$ if and only if there exists a bag of the decomposition 
containing both $x$ and $y$. A set $W \subseteq V$ of vertices induces a clique in $H$ if and only if there is a bag $X_i \in {\cal X}$ such that $W \subseteq X_i$.
\end{proposition}
It is well-known that the graph $H$ constructed above is a chordal graph (or an interval graph if we replace tree decomposition by path decomposition), but we will not use this here. 
See e.g.~\cite{Gol80} for more details on these graphs and a proof of the previous proposition.

Let $i$ be a node of an arbitrarily rooted tree decomposition $(T,{\cal X})$. Let $T_i$ be the subtree of $T$ rooted in $i$. We denote by $V_i$ the union of bags of the subtree $T_i$. We let $L_i = V_i \setminus X_i$ ($L$ like ``lower'') and $R_i = V \setminus V_i$ ($R$ like ``rest''). Clearly, $(L_i,X_i,R_i)$ is a partition of $V$. 

\begin{proposition}[\cite{Bodlaender97}]\label{pr:sep}
Let $(T, {\cal X})$ be a tree decomposition of graph $G=(V,E)$. The bag $X_i$ separates, in graph $G$, any two vertices $a \in L_i$ and $b \in R_i$, i.e. $a$ and $b$ are in different components of $G[V \setminus X_i]$.
\end{proposition}

For our purpose, it is very convenient to use \emph{nice} tree and path decompositions (see e.g.~\cite{Bodlaender97}). In a nice tree decomposition $(T, {\cal X})$, the tree  is rooted, and has only four types of nodes~:
\begin{enumerate}
\item Leaf nodes $i$, in which case $|X_i| = 1$.
\item Introduce nodes $i$, having a unique child $j$ s.t. $X_i = X_j \cup \{u\}$ for some $u \in V \setminus X_j$.
\item Forget nodes $i$, having a unique child $j$ s.t. $X_i = X_j \setminus \{u\}$ for some $u \in  X_j$.
\item Join nodes $i$, having exactly two children $j$ and $k$, s.t. $X_i = X_j = X_k$.
\end{enumerate}
Moreover, we can assume that the root node corresponds to a bag of size 1.

Let us associate an \emph{operation} $\tau_i$ to each node of a nice tree decomposition. If we are in the second case of the definition (introduce node $i$), we associate operation $\tau_i = introduce(u)$, where $u$ is the vertex introduced in bag $X_i$. If we are in the third case (forget node $i$), we associate operation $\tau_i = forget(u)$, where $u$ is the forgotten vertex.
In the fourth case (join node), we associate operation $\tau_i = join(X_i; L_j, L_k)$. For a leaf node $i$ with $X_i=\{u\}$, we also associate operation $\tau_i = introduce(u)$. 
Nice path decompositions are defined in a similar way, but of course they do not have join nodes. 

It is well known~\cite{Bodlaender97} that any tree or path  decomposition can be refined into a nice one in linear time, without increasing the width. 

\begin{proposition}[\cite{Bodlaender97}]\label{pr:nice}
Let $(T,{\cal X})$ be a tree decomposition of $G$. There exists a nice tree decomposition $(T',{\cal X'})$, such that 
\begin{itemize}
\item each bag of ${\cal X'}$ is a subset of a bag in ${\cal X}$
\item for each node $i$ of $T$, there is a node $i'$ of $T'$ such that the corresponding partitions $(L_i,X_i,R_i)$ (induced by $i$ in $(T,{\cal X})$) and  $(L'_{i'},X'_{i'},R'_{i'})$ (induced by $i'$ in $(T',{\cal X'})$) are equal.
\end{itemize}
\end{proposition}
\subsubsection*{Traces and valid partitions}
Let  $C$ be a vertex cover of minimum size of our input graph $G$, and let $S = V\setminus C$ be the remaining independent set. We denote $k = |C|$. Our objective is to describe,
in a first step, an $O^*(4^k)$ algorithm for treewidth and an $O^*(3^k)$ algorithm for pathwidth. Very informally, if we fix an nice tree or path decomposition of $G[C]$, then there is an optimal way of adding the vertices of $S$ to this tree or path decomposition. Trying all nice decompositions of $G[C]$ by brute force would be too costly. Therefore we introduce the notion of \emph{traces} and \emph{valid partitions} of $C$.

\begin{definition}\label{de:trace}
Consider a node $i$ of a tree decomposition $(T,{\cal X})$ of $G$. The \emph{trace} of node $i$ on $C$ is the three-partition $(L_i^C, X_i^C, R_i^C)$ of $C$ such that $L_i^C=L_i \cap C$, $X_i^C=X_i \cap C$ and $R_i^C = R_i \cap C$.

A partition $(L^C,X^C,R^C)$ of $C$ is called a \emph{valid triple} or \emph{valid partition} if it is the trace of some node of a tree decomposition. We say that a tree decomposition \emph{respects} the valid partition  $(L^C,X^C,R^C)$ if some node of the tree decomposition produces this trace on $C$. 
\end{definition}

The following lemma gives an easy characterization of valid partitions of $C$. It also proves that a partition is the trace of a node of some tree decomposition, this also holds for some path decomposition. Therefore we do not need to distinguish between partitions that would be valid for tree decompositions or valid for path decompositions. 
\begin{lemma}\label{le:valid}
A three-partition $(L^C,X^C,R^C)$ is the trace of some tree decomposition (or path decomposition) if and only if $X^C$ separates $L^C$ from $R^C$ in the graph $G[C]$. 
\end{lemma}
\begin{proof}
``$\Rightarrow$:'' Consider a node $i$ of a tree decomposition $(T,{\cal X})$ of $G$ such that $(L^C,X^C,R^C)$ is the trace of node $i$ on $C$. By Proposition~\ref{pr:sep}, bag $X_i$ separates $L_i$ from $R_i $ in $G$. Therefore $X_i \cap C = X^C$ separates $L_i \cap C = L^C$ from $R_i \cap C = R^C$ in $G[C]$.

`$\Leftarrow$:'' Conversely, since $X^C$ separates $L^C$ from $R^C$ in  $G[C]$ and $S = V \setminus C$ is an independent set of $G$, note that the three bags $L^C \cup S, X^C \cup S$ and $R^C \cup S$ form a path decomposition of $G$. The trace of the middle bag is $(L^C,X^C,R^C)$.
\qed
\end{proof}

By Proposition~\ref{pr:nice}, for any valid partition $(L^C,X^C,R^C)$, there exists a nice tree or path decomposition respecting it.
Our algorithms will proceed by dynamic programming over valid three-partitions $(L^C,X^C,R^C)$ of this type, for a given vertex cover $C$. There is a natural partial ordering on such three-partitions.

\begin{definition}\label{de:prec}
A valid three-partition $(L_j^C,X_j^C,R_j^C)$ \emph{precedes} the three-partition $(L_i^C,X_i^C,R_i^C)$ if they are different and they are the respective traces of two nodes $j$ and $i$ of a same nice tree decomposition $(T,{\cal X})$, where $i$ is the father of $j$ in $T$.
\end{definition}

Observe that if $(L_j^C,X_j^C,R_j^C)$ precedes $(L_i^C,X_i^C,R_i^C)$ we have that $L_j^C \subsetneq L_i^C$ (if $i$ is a $join$ or $forget$ node) or $L_j^C =L_i^C$ and $X_j^C \subsetneq X_i^C$ (if $i$ is an $introduce$ node). In particular, we can order the three-partitions according to a linear extension of the precedence relation.
Our algorithms will proceed by dynamic programming over three-partitions of $C$, according to this order. 

It is convenient for us to have a unique maximal three-partition w.r.t. the precedence order. Therefore, starting from graph $G$, we create a new graph $G'$ by adding a universal vertex $univ$ (i.e. adjacent to all other vertices of $G$). Clearly, $C \cup \{univ\}$ is a vertex cover of $G'$, of size $k+1$. Note that the treewidth (resp. pathwidth) of $G'$ equals the treewidth (resp. pathwidth) of $G$, plus one. Moreover, $G$ has an optimal nice tree (resp. path) decomposition whose root bag only contains vertex $univ$. Therefore, it is sufficient to compute the treewidth (pathwidth) for graph $G'$. From now on we assume that the input graph is $G'$, i.e. it contains a special universal vertex $univ$, and we only use nice tree (path) decompositions whose root bag is $\{univ\}$. If $C$ denotes the vertex cover of the input graph, then the trace of the root is always $(C \sm \{univ\}, \{univ\},\emptyset)$.

\section{\textsc{Treewidth} parameterized by vertex cover}\label{se:tw4k}

%
Recall that the nice tree decompositions are rooted, thus we can speak of \emph{lower} and \emph{upper} nodes of the decomposition tree.

\begin{lemma}\label{le:path}
Let $(L^C,X^C,R^C)$ be a three-partition of $C$. Let $(T,{\cal X})$ be a nice
tree-decomposition and consider the set of nodes of $T$ whose trace on $C$
is  $(L^C,X^C,R^C)$. 
%
If $L^C \neq \emptyset$, then these nodes of $T$ induce a directed subpath in $T$, 
from a lower node $imin$ to an upper node $imax$. 
\end{lemma}
\begin{proof}
Consider two nodes $i$ and $j$ leaving this same trace $(L^C,X^C,R^C)$ on $C$. We claim that one of them is ancestor of the other in the tree. By contradiction, assume there is a lowest common ancestor $k$ of $i$ and $j$, different from $i,j$. Let $x \in L^C$ (note that here we use the condition $L^C \neq \emptyset$). Observe that $x$ appears in bags of both subtrees $T_i$ and $T_j$ of $T$, hence by definition of a tree decomposition it must belong to bag $X_k$. Since $x$ is in $X_k$ and in the subtree $T_i$, we must also have $x \in X_i$. But $X_i \cap C = X^C$, implying that $x$ is in both $X^C$ and $L^C$ --- contradicting the fact that the latter sets do not intersect. It follows that one of $i,j$ must be ancestor of the other. 

Let $imin$ (resp. $imax$) be the lowest (resp. highest) node whose trace on $C$ is $(L^C,X^C,R^C)$. 
It remains to prove that any node on the path from $imin$ to $imax$ in $T$ leaves the same trace. Let $i$ be a node on this path. Recall that $L_i$ denotes the set of vertices of $G$ that appear only in bags strictly below $i$, and $R_i$ denotes the vertices that do not appear in bags below $i$. Since $i$ is between $imin$ and $imax$, cleary $X^C \subseteq X_i$. If $X_i$ contains some vertex $x \in C \sm X^C$, then either $x \in L^C$ and thus $x$ must also appear in bag $X_{imin}$, or $x \in R^C$ and it must appear in bag $X_{imax}$. In both cases, this contradicts the trace of $imin$ and $imax$ on $C$. We thus have $L^C \subseteq L_i \cap C$. If $L_i \cap C$ also contains some node $x \in R^C$, as before we have that $x$ must be in bag $X_{imax}$ --- a contradiction. Eventually, observe that $R^C \subseteq R_i \cap C$, and that if
$R_i \cap C$ contained some vertex $x \in L^C$, this vertex must appear in bag $imin$ --- a contradiction.
\qed
\end{proof}

In order to ``glue'' a valid three-partition $(L^C,X^C,R^C)$ with the previous and next ones, into a nice tree decomposition of $G[C]$, we need to control the operation right below and right above the subpath of nodes leaving this trace. Therefore we introduce the following notion of \emph{valid quintuples}. 

\begin{definition}\label{de:tvalidq}
Let $(L^C, X^C,R^C)$ be a valid partition of $C$, with $L^C \neq \emptyset$. Let $\tau_+$ and $\tau_-$ be operations of type $introduce$, $forget$ or $join$. We say that $(\tau_-,L^C,X^C,R^C,\tau_+)$ is
a \emph{valid quintuple} if there is a nice tree decomposition  $(T,{\cal X})$ of $G$ respecting $(L^C, X^C,R^C)$, with  $imin$ and $imax$ being the lower and upper node corresponding to this trace,
such that $\tau_- = \tau_{imin}$ and $\tau_+ = \tau_{imax+1}$, where $imax+1$ is the father of $imax$. In the particular case when $imax$ is the root we assume for convenience that $\tau_{imax+1}$ is the $forget$ operation on the unique vertex of the root bag.

We say that this nice tree decomposition $(T,{\cal X})$  \emph{respects} the quintuple $(\tau_-,L^C,X^C,R^C,\tau_+)$.
\end{definition}

The following result characterizes all valid quintuples and will be used by our algorithm to enumerate all of them. Its proof is moved to the Appendix, Subsection~\ref{ss:le:tvalidq}.
\begin{lemma}\label{le:tvalidq}
A quintuple $Q= (\tau_-,L^C,X^C,R^C,\tau_+)$ is valid if and only if $(L^C, X^C,R^C)$ is a valid partition of $C$, $L^C$ is not empty  and the following two conditions hold: 
\begin{itemize}
\item $\tau_-$ is of type 
\begin{itemize}
\item $introduce(u)$ for some $u \in X^C$ and $N(u) \cap L^C = \emptyset$, or
\item $forget(u)$ for some $u \in L^C$, or
\item $join(X^C;L1^C,L2^C)$ where $(L1^C,L2^C)$ forms a non-trivial two-partition of $L^C$ and the three-partitions $(L1^C,X^C,R^C \cup L2^C)$ and $(L2^C,X^C,R^C \cup L1^C)$ are valid.
\end{itemize}
\item $\tau_+$ is of type
\begin{itemize}
\item $introduce(v)$ with $v \in R^C$, or
\item $forget(v)$ with $v \in X^C$ and $N(v) \cap R^C = \emptyset$, or
\item $join(X^C; L^C,LR^C)$ where $LR^C$ is a non-empty subset of $R^C$, and three-partitions $(L^C \cup LR^C, X^C, R^C \setminus LR^C)$ and $(LR^C, X^C, R^C\setminus LR^C \cup L^C)$ are valid.
\end{itemize}
\end{itemize}
\end{lemma}

%
%

\medskip
To be able to start our dynamic programming, we introduce a new category of valid quintuples that we call \emph{degenerate}, corresponding to valid partitions of type $(\emptyset,X^C,R^C)$. Roughly, they will correspond to the leaves of our optimal tree decomposition. 
We point out that for degenerate quintuples, parameter $\tau_-$ is irrelevant.

\begin{definition}
Let $(\emptyset,X^C,R^C)$ be a valid partition of $C$ and let $\tau_+$ be an operation of type $forget(u)$, with $u \in X^C$ such that $N_G(u) \subseteq X^C$. We say that $(\tau_-,\emptyset,X^C,R^C,\tau_+)$ is a \emph{degenerate valid quintuple} and a tree decomposition \emph{respects} this quintuple if it has a node $imax$ whose trace on $C$ is $(\emptyset,X^C,R^C)$, and whose father corresponds to operation $forget(u)$.
\end{definition}

Let us fix a valid quintuple $Q = (\tau_-,L^C,X^C,R^C,\tau_+)$. We want to construct a tree decomposition $(T,{\cal X})$ respecting $Q$, of minimum width. We must understand how to place the vertices of $S$ in the bags of $(T,{\cal X})$. For this purpose we define some  special subsets of $S$ w.r.t. $Q$, and the next lemmata describe how these subsets are forced to be in some bags on the subpath of $T$ from $imin$ to  $imax$ (cf. Lemma~\ref{le:path}).

\begin{notation}\label{no:tsets}
Let $Q = (\tau_-,L^C,X^C,R^C,\tau_+)$ be a valid quintuple. 
\begin{itemize}
\item We denote $XTR^S(Q) = \{x \in S \mid N(x) \cap L^C \neq \emptyset \text{~and~} N(x) \cap R^C \neq \emptyset\}$.
\item 
\begin{itemize}
\item If $\tau_-$ is of type $introduce(u)$, then we denote $XL^S(Q) = \{x \in S \mid N(x) \subseteq L^C \cup X^C \text{~and~} u \in N(x) \text{~and~} N(x) \cap L^C \neq \emptyset\}$.
\item If $\tau_-$ is of type $join(X^C;L1^C,L2^C)$, then $XL^S(Q) = \{x \in S \mid N(x) \cap L1^C \neq \emptyset  \text{~and~} N(x) \cap L2^C \neq \emptyset   \text{~and~}  N(x) \cap R^C = \emptyset\}$. In particular, the last condition ensures that $XL^S(Q)$ does not intersect $XTR^S(Q)$.
\item If $\tau_-$ is a $forget$ operation or if the quintuple is degenerate, then we let $XL^S(Q) = \emptyset$.
\end{itemize}
\item Suppose that $\tau_+$ is of type $forget(v)$. Then we let $XR^S(Q) = \{x \in S \mid  N(x) \subseteq R^C \cup X^C \text{~and~} v \in N(x) \text{~and~} N(x) \cap R^C \neq \emptyset\}$.
If $\tau_+$ is a $introduce$ or $join$ operation, then $XR^S = \emptyset$.
\end{itemize}
\end{notation}

\begin{lemma}\label{le:timinimax}
Let $Q = (\tau_-,L^C,X^C,R^C,\tau_+)$ be a valid quintuple and let $(T, {\cal X})$ be a nice tree decomposition respecting $Q$. Denote by $[imin,imax]$ the directed subpath of nodes whose trace on $C$ is $(L^C,X^C,R^C)$ (in the case where $L^C = \emptyset$, we take $imin=imax$). Then
\begin{itemize}
\item For any $i$ in the subpath $[imin,imax]$, $X_i$ contains $X^C \cup XTR^S(Q)$.
\item $X_{imin}$ contains $X^C \cup XTR^S(Q) \cup XL^S(Q)$.
\item $X_{imax}$ contains $X^C \cup XTR^S(Q) \cup XR^S(Q)$.
\end{itemize}
\end{lemma}
\begin{proof}
Let $x \in XTR^S(Q)$. By definition of $XTR^S$, vertex $x$ has a neighbor $a \in L^C$ and a neighbor $b \in R^C$. Since $a \in L^C$, it only appears in the bags of $T$ strictly below $imin$. Since $x$ is adjacent to $a$, it must also appear on one of these bags. Since $b \in R^C$, vertex $b$ appears in no bag below $imax$ (included). Therefore, $x$ must appear in some bag which is not below $imax$. Consequently, $x$ appears in every bag of the $[imin,imax]$ subpath.

Assume that $XL^S(Q)$ is not empty. If $\tau_- = \tau_{imin} = introduce(u)$, then every vertex $x \in XL^S(Q)$ must appear in some bag strictly below $imin$ (because it has a neighbor in $L^C$) and in some bag containing $u$ (because it sees $u$). This latter bag cannot be strictly below $imin$. Thus $x \in X_{imin}$ and $XL^S(Q)$ is contained in $X_{imin}$. 
When $imin$ is a $join$ node, $L^C \neq \emptyset$ and we must show that $X_{imin}$ contains $XL^S(Q)$. But then each vertex $x \in XL^S(Q)$ has a neighbor which only appears in the left subtree of $imin$, strictly below $imin$, and one in the right subtree of $imin$, strictly below $imin$. Thus $x$ must appear in the bag of $imin$.

If $XR^S(Q)$ is not empty, then $\tau_+ = forget(v)$ and $v$ is a neighbor of each $x \in XR^S(Q)$. Hence $x$ must appear in a bag below $imax$ (included). But $x$ also has neighbors in $R^C$, thus it must appear in some bag which is not below $imax$. Consequently, $XR^S(Q)$ is contained in $X_{imax}$. 
\qed
\end{proof}

We consider now vertices of $S$ whose neighborhood is a subset of $X^C$.

\begin{notation}\label{no:tXF}
Let $Q = (\tau_-,L^C,X^C,R^C,\tau_+)$ be a valid quintuple. We denote by $XF^S(Q)$ the set  of vertices $x\in S$ such that $N(x) \subseteq X^C$.
 
Let $\epsilon(Q)$ be set to $1$ if there is some $x \in XF^S(Q)$ such that $N(x)=X^C$, set to $0$ otherwise.
\end{notation}

\begin{lemma}\label{le:tepsilon}
Let $(T,{\cal X})$ be a tree decomposition respecting a quintuple $Q = (\tau_-,L^C,X^C,R^C,\tau_+)$. Then $(T,{\cal X})$ has a bag of
size at least $|X^C| + \epsilon(Q)$.
\end{lemma}
\begin{proof}
If $\epsilon(Q) = 0$ the claim is trivial. If $\epsilon(Q) = 1$, let $x \in S$ such that $N(x) = X^C$. By Helly's property (see Proposition~\ref{pr:hellyc}), there must be a bag
of $(T,{\cal X})$ containing $x$ and $X^C$.
\qed
\end{proof}

\begin{notation}\label{no:loctw}
Let $Q = (\tau_-,L^C,X^C,R^C,\tau_+)$ be a valid quintuple. We define the \emph{local treewidth} of $Q$ as
$$loctw(Q) =|X^C| +\max\{|XTR^S|  + |XL^S|, |XTR^S| + |XR^S|, \epsilon(Q)\} - 1.$$
\end{notation}

The $-1$ used above plays the same role as in the definition of treewidth. By Lemmata~\ref{le:timinimax} and~\ref{le:tepsilon} we deduce.
\begin{corollary}\label{co:loctw}
Any nice tree decomposition of $G$ respecting a valid quintuple $Q$ is of width at least $loctw(Q)$.
\end{corollary}

We now define the partial treewidth of a valid quintuple. Intuitively, the partial treewidth of a quintuple $Q$ is the minimum value $t$ such that there is a nice tree decomposition of $G[C]$, respecting $Q$, with all valid quintuples below $Q$ having local treewidth at most $t$. We shall prove in Lemma~\ref{le:geptw} and Theorem~\ref{th:tw} that actually the partial treewidth of $Q$ is at most $t$ if and only if there exists a nice tree decomposition of the whole graph $G$, respecting $Q$, such that all  bags below $imax$ have size at most $t+1$.

\begin{notation}\label{no:ptw}
Given a valid quintuple $Q= (\tau_{-},L^C,X^C,R^C,\tau_{+})$, we define the \emph{partial treewidth} of $Q$, denoted $ptw(Q)$, as follows.

\begin{itemize}
\item If $Q = (\tau_-,\emptyset, X^C, R^C, forget(u))$ is a degenerate valid quintuple then
$$ptw(Q) = loctw(Q).$$
\item If $\tau_{-} = introduce(u)$,
 $$ptw(Q) = \max\left\{loctw(\tau_{-},L^C,X^C,R^C,\tau_{+}), \min_{\text{valid quintuple~} Q_-} ptw(Q_-)\right\}$$

where the minimum is taken over all valid quintuples $Q_-$ of type $(\tau_{--},L^C,X^C \sm \{u\},R^C \cup \{u\}), introduce(u))$.
\item If $\tau_{-} = forget(u)$,
 $$ptw(Q) = \max\left\{loctw(\tau_{-},L^C,X^C,R^C,\tau_{+}), \min_{\text{valid quintuple~} Q_-} ptw(Q_-)\right\}$$

where the minimum is taken over all valid quintuples $Q_-$ of type $(\tau_{--},L^C\setminus\{u\},X^C \cup \{u\},R^C),$ $forget(u))$.
\item If $\tau_{-} = join(X^C;L1^C,L2^C)$,
\begin{eqnarray*}
 ptw(Q) = &\max&(loctw(\tau_{-},L^C,X^C,R^C,\tau_{+}), \\ 
        && \min_{\text{valid quintuple~} Q1_-} ptw(Q1_-),\\
        && \min_{\text{valid quintuple~} Q2_-} ptw(Q2_-))
\end{eqnarray*}
where the minima are taken over all valid quintuples $Q1_-$ of type $(\tau1_{--},L1^C,X^C,R^C \cup L2^C),join(X^C;L1^C,L2^C))$ and
all quintuples $Q2_-$ of type $(\tau2_{--},L2^C,X^C,R^C \cup L1^C),$ $join(X^C;L1^C,L2^C))$.
\end{itemize}
\end{notation}

\begin{lemma}\label{le:geptw}
Any nice tree decomposition of $G$ respecting a valid quintuple $Q$ is of width at least $ptw(Q)$.
\end{lemma}
\begin{proof}
We order the three-partitions $(L^C,X^C,R^C)$ of $C$ according to the precedence relation (Definition~\ref{de:prec}).
We prove the lemma for every valid quintuple $Q=(\tau_-,L^C,X^C,R^C,\tau_+)$, by induction (according to this order) on $(L^C,X^C,R^C)$. 

For quintuples such that $L^C = \emptyset$, the property follows directly from Corollary~\ref{co:loctw} and the base case of Notation~\ref{no:ptw}. 

Now take $Q=(\tau_-,L^C,X^C,R^C,\tau_+)$ with $L^C \neq \emptyset$. Let $imin$ the lowest node of the tree decomposition respecting $Q$, whose trace is 
$(L^C,X^C,R^C)$. If $imin$ is a join node, it has two sons with traces $(L1^C,X^C,R^C \cup L2^C)$ and $(L2^C,X^C,R^C \cup L1^C)$ and the proof follows from the $join$ case of Notation~\ref{no:ptw} and the induction hypothesis on the valid quintuples preceding $Q$. Note that both $L1^C$ and $L2^C$ are non empty, otherwise $imin$ would not be the lowest node with trace 
$(L^C,X^C,R^C)$. Similarily, if $imin$ is an $introduce(u)$ node, then we apply Corollary~\ref{co:loctw} and the $introduce$ case of Notation~\ref{no:ptw}
to the quintuple preceding $Q$ in the tree decomposition. The same holds if $imin$ is of type $forget(u)$ (using the $forget$ case of Notation~\ref{no:ptw}). 
We point out that, if $\tau_- = forget(u)$ and $L^C=\{u\}$, the quintuple $Q_-$ of Notation~\ref{no:ptw} corresponds to the base case of our induction.
\qed
\end{proof}

The next theorem is the main combinatorial tool for our \textsc{Treewidth} algorithm.

\begin{theorem}\label{th:tw}
The treewidth of $G$ is
	$$tw(G) = \min_{Q_{\text{last}}} ptw(Q_{\text{last}})$$
over all valid quintuples $Q_{\text{last}}$ of the form $(\tau_-,C \setminus \{univ\}, \{univ\}, \emptyset, forget(univ))$.
\end{theorem}
\begin{proof}
First note that $tw(G) \geq \min_{Q_{\text{last}}} ptw(Q_{\text{last}})$. Indeed, an optimal tree decomposition will contain a root whose bag corresponds to a single vertex $univ$, and this root will leave a trace on $C$ of type $(C \setminus \{univ\}, \{univ\}, \emptyset)$. The inequality follows from Lemma~\ref{le:geptw}.

Conversely, let $Q_{\text{last}}=(\tau_-,C \setminus \{univ\}, \{univ\}, \emptyset, forget(univ))$ be the valid quintuple of minimum $ptw$, among all quintuples of this type; denote by $t$ this minimum value. The computation of $ptw(Q_{\text{last}})$ naturally provides a tree $T^C$ of quintuples, the root being $Q_{last}$, and such that for the node corresponding to quintuple $Q$ its sons are the preceding quintuples realizing the minimum value for $ptw(Q)$ in Notation~\ref{no:ptw}. The leaves of this tree correspond to the base case of Notation~\ref{no:ptw}, hence to degenerate valid quintuples. 
By definition of $ptw$, all these selected quintuples have $loctw$ at most $t$. We construct a tree decomposition of $G$ with bags of size at most $t+1$.

 Let $Q_i =( \tau_{-i}, L^C_i, X^C_i, R^C_i, \tau_{+i})$ be the quintuple associated to node $i$ in $T^C$. Let $(T^C, {\cal X}^C)$ be the tree-decomposition of $G[C]$ obtained by associating to each node $i$ of $T^C$ the bag $X^C$. Each node $i$, except for the leaves, corresponds to an $introduce$, $forget$ or $join$ operation $\tau_{-i}$.
  
Let $T$ be the tree obtained from $T^C$ by replacing each node $i$ with a path of three nodes, denoted $imin, imid$ and $imax$ (from the bottom towards the top). 
Initially, we associate to the three nodes $imin, imid, imax$ the same bag $X^C_i$. Now, for each $i$,
\begin{enumerate}
\item add $XTR^S(Q_i)$ to all bags in the subpath $[imin,imax]$ of $T$;
\item add $XL^S(Q_i)$ to bag number $imin$;
\item add $XR^S(Q_i)$ to bag number $imax$;
\item For each vertex $x \in XF^S(Q_i)$, which has not yet been added to some bag of $T$, create a new node of $T$ adjacent only to $imid$ and associate to this node the bag $N[x]$. These nodes are called \emph{pending nodes}.
\end{enumerate}

We claim that in this way we have obtained a tree decomposition $(T, {\cal X})$ of $G$. Clearly all bags created at step $i$ are of size at most $loctw(Q_i)+1$, hence at most $t+1$. It remains to prove that these bags satisfy the conditions of a tree decomposition.

Recall that $(T^C, {\cal X}^C)$ is a tree decomposition of $G[C]$. By construction of $(T, {\cal X})$,  for each vertex $y \in C$, the bags of $(T, {\cal X})$ containing it will form a subtree of $T$. Also, for each edge $yz$ of $G[C]$, some bag of $(T, {\cal X})$ shall contain both $y$ and $z$. It remains to verify the same type of conditions for vertices of $S$ and edges incident to them. 
This part of the proof is skipped due to space restrictions, see details in the Appendix, Subsection~\ref{ss:th:tw}.
\qed
\end{proof}

\begin{theorem}\label{th:twvc}
The \textsc{Treewidth} problem can be solved in $O^*(4^k)$ time, where $k$ is the size of the minimum vertex cover of the input graph. 
\end{theorem}
\begin{proof}
Given an arbitrary graph $G$, we first compute a minimum vertex cover in $O^*(1.28^k)$ (Proposition~\ref{pr:VC}). Then $G$ is transformed into a graph $G'$ by adding a universal vertex $univ$. Let $C$ be the vertex cover of $G'$ obtained by adding $univ$ to the minimum vertex cover of $G$ (hence $|C| = k+1$). The treewidth of $G'$ is computed as follows.
\begin{enumerate}
\item Compute all valid partitions $(L^C,X^C,R^C)$, by enumerating all three-partitions of $C$ and keeping only the valid ones (Lemma~\ref{le:valid}). This can be done in time $O^*(3^k)$. The number of valid partitions is at most $3^{k+1}$.
\item Compute all valid quintuples using Lemma~\ref{le:tvalidq}. For quintuples 
$Q=(\tau_{-},L^C,X^C,R^C,\tau_{+})$ where $\tau_+$ is a $join$ node, the parameters of
this $join$ are not relevant for $loctw(Q)$ and $ptw(Q)$. Therefore, we do not need to memorize the parameters of the $join$. With this simplification, we only need to store $O^*(4^k)$ valid (simplified) quintuples, and their computation can be performed in time $O^*(4^k)$. The $4^k$ comes from quintuples of the type $(join(X^C;L1^C,L2^C),L^C,$ $X^C,R^C,\tau_+)$, since 
$(L1^C, L2^C, X^C,R^C)$ is a partition of $C$ into four parts. The quintuples are then sorted by the precedence relation on the corresponding valid three-partitions. This  can be done within the same running time, the triples $(L^C,X^C,R^C)$ being sorted by increasing size of $L^C$, and in case of tie-breaks by increasing size of $X^C$ (see Definition~\ref{de:prec} and following remarks).
\item For each valid quintuple $Q=(\tau_{-},L^C,X^C,R^C,\tau_{+})$, according to the ordering above, compute by dynamic programming $loctw(Q)$ (Notation~\ref{no:loctw}) and then $ptw(Q)$ (Notation~\ref{no:ptw}). In order to process efficiently the quintuples $Q_-$  of Notation~\ref{no:ptw}, let us  observe the value $\min_{Q_-} ptw(Q_-)$ over all $Q_-$ of a given type can be updated online, as soon as we compute $ptw(Q_-)$. Indeed, for all the $Q_-$ of a same type, the first parameter $\tau_{--}$ will differ, but the four others are equal. So it is actually a minimum over all $\tau_{--}$. 
The same holds for the minimum over all $Q1_-$ and over all $Q2_-$. Hence, when we process quintuple $Q$, we have these minima at hand and the value $ptw(Q)$ is computable in polynomial time. This step can be performed in polynomial for each $Q$, so the overall running time is still $O^*(4^k)$.
\item Compute the treewidth of $G'$ using Theorem~\ref{th:tw}, and return $tw(G)=tw(G')-1$. This step takes polynomial running time. 
\end{enumerate}
Altogether, the algorithm takes $O^*(4^k)$ running time and space. This achieves the proof of the theorem. The algorithm can also be adapted to return, within the same time bounds, an optimal tree decomposition of the input graph.
\qed
\end{proof}

Note that the algorithm for pathwidth, described in Appendix~\ref{se:pw} is quite similar, with a slight difference in the definition of local pathwidth (Notation~\ref{no:bagsize}). Due to the fact that it only uses $introduce$ and $forget$ operations, the number of valid quintuples is $O^*(3^k)$, and so is the running time of the pathwidth algorithm.
\begin{theorem}
The \textsc{Pathwidth} problem can be solved in $O^*(3^k)$ time, where $k$ is the size of the minimum vertex cover of the input graph. 
\end{theorem}

\section{An $O^*(3^k)$ algorithm for \textsc{treewidth} (sketch)}
We want to improve the running time of our algorithm for treewidth from $O^*(4^k)$ to $O^*(3^k)$. Due to space restrictions we only sketch here, rather informally, the main ideas. Full details are given in Appendix~\ref{se:tw3k}.

We need to cope with join quintuples $Q= (\tau_- = join(X^C;L1^C,L2^C),L^C,X^C,R^C,\tau_+)$ 
because we cannot afford to store parameters $L1^C,L2^C$; we only want to recall that $\tau_-$ is of type $join$. (Note that our algorithm already does this kind of simplification when $\tau_+$ is of type $join$).

Recall that the partial treewidth $ptw(Q)$ of a valid quintuple $Q$ is defined (Notation~\ref{no:ptw}) as the minimum value $t$ such that there is a \emph{trimmed} nice tree decomposition of $G[C]$, with a node $i$ corresponding to valid quintuple $Q$, and such that all nodes below $i$ have valid quintuples of local treewidth at most $t$. (By \emph{trimmed} nice decomposition we mean that leaf bags might not be of size 1, see Appendix~\ref{se:tw3k} for a formal description.) The part of the tree decomposition rooted in $i$ is called a $Q$-rooted subtree decomposition. We introduce an integer parameter $d$ and define $ptw(d,Q)$ as above, but adding the constraint that, in the $Q$-rooted subtree decomposition, every path from a leaf to the root has at most $d$ nodes of type $join$. E.g., if $d = 0$, it means that the $Q$-rooted subtree must be a path. Observe that, for $d=k$, we have $ptw(k,Q) = ptw(Q)$, because any nice decomposition of $G[C]$ will have at most $k$ join nodes from a leaf to the root.

For ``simplified'' join quintuples $\tilde{Q} = (join,L^C,X^C,R^C,\tau_+)$ we let $ptw(d,\tilde{Q})$ be the minimum value of $ptw(d, join(X^C;L1^C,L2^C),X^C,R^C,\tau_+)$, over all possible partitions $(L1^C,L2^C)$ of $L^C$.

Our algorithm \textsc{Ptw} (Algorithm~\ref{al:Ptw}, Appendix~\ref{se:tw3k}) computes  in $O^*(3^k)$ time the values $pwd(d,\tilde{Q})$ for all simplified quintuples $\tilde{Q}$, from values $ptw(d-1,\dots)$. For that purpose, it first calls
a specific algorithm \textsc{JoinPtw}, which computes in $O^*(3^k)$ time all values $ptw(d,(join,\dots))$ for all simplified join quintuples. For all other quintuples, it simply runs (almost) the same algorithm as in Theorem~\ref{th:twvc}, and since we avoid joins this can be done within the required time bound of  $O^*(3^k)$.

We give some hints about the \textsc{JoinPtw} algorithm (Algorithm~\ref{al:JoinPtw}, Appendix~\ref{se:tw3k}), computing all values $ptw(d,(join,\dots))$ for simplified join quintuples. Let $\tilde{Q} = (join,L^C,X^C,R^C,\tau_+)$. To check that $ptw(d, \tilde{Q}) \leq t$, we need the existence of a two-partition $(L1^C, L2^C)$ of $L^C$ and of two (simplified) quintuples $Q1_- =(\tau1_{--},L1^C,X^C, C \setminus (X^C \cup  L1^C),join)$ and 
$Q1_- =(\tau2_{--}, L2^C,X^C, C \setminus (X^C \cup  L2^C),join)$ such that $ptw(d-1, Q1_-)\leq t$ and $ptw(d-1, Q2_-)\leq t$. Observe that here we make use of parameter $d$. Now comes into play the fast subset convolution (Theorem~\ref{th:convolution}).  We fix a set $X^C$. We define a boolean function $ptw\_atmost\_t(L1^C)$ over subsets $L1^C$ of $C \setminus X^C$, which is set to $1$ if $ptw(d-1,(\tau1_{--},L1^C,X^C, C \setminus (X^C \cup  L1^C),join)) \leq t$ for some $\tau1_{--}$, set to 0 otherwise. 
The subset convolution $(ptw\_atmost\_t \star ptw\_atmost\_t)$ is an integer function over subsets $L^C$ of $C \setminus X^C$ defined as:
$$(ptw\_atmost\_t \star ptw\_atmost\_t)(L^C) = \sum_{(L1^C,L2^C)} ptw\_atmost\_t(L1^C) \cdot ptw\_atmost\_t(L2^C)$$
over all partitions  $(L1^C,L2^C)$ of $L^C$.
Therefore, the existence of $Q1_-$ and $Q2_-$ is equivalent to the fact that $(ptw\_atmost\_t \star ptw\_atmost\_t)(L^C) \geq 1$. By Theorem~\ref{th:convolution}, this subset convolution, over all subsets $L^C$ of $C \setminus X^C$, can be computed in $O^*(2^{k - |X^C|})$. Hence, for a fixed $X^C$, all quintuples $\tilde{Q} = (join,L^C,X^C,R^C,\tau_+)$ can be processed within the same running time. This will make $O^*(3^k)$ time over all simplified join quintuples. Several technical points have been deliberately omitted, especially the fact that we cannot define $loctw(\tilde{Q})$ on a simplified join quintuple, because the set $XL^S$ (Notation~\ref{no:tsets}) depends on the parameters of the $join$. Full details are in Appendix~\ref{se:tw3k}.

\section{Concluding remarks}

We have shown that it is possible to obtain $O^*(3^k)$ time algorithms for  computing 
{\sc treewidth} and {\sc pathwidth} where parameter $k$ is the size of the {\it vertex cover} 
of the graph.  This puts {\it vertex cover} in the same class as parameter $n$ 
as both allows an $O^*(c^k)$ time algorithm for the considered problems. 
It is an interesting question whether an $O^*(c^k)$ time algorithm exists when using 
{\it feedback vertex set} of the graph as the parameter $k$.

\begin{small}

\bibliographystyle{plain}
\bibliography{LayoutVC}
 
\end{small}

\appendix

\section{\textsc{Pathwidth} parameterized by vertex cover}\label{se:pw}

In the case of pathwidth, the following lemma holds even if $L^C = \emptyset$.
\begin{lemma}\label{le:psubpath}
Let $(L^C,X^C,R^C)$ be a three-partition of $C$ and let $(P,{\cal X})$ be a nice path decomposition of $G$ respecting it. The set of nodes of $P$ whose trace is $(L^C,X^C,R^C)$  induces a (connected) subpath of $P$, from a lower node $imin$ to an upper node $imax$. 
\end{lemma}
\begin{proof}
The arguments of Lemma~\ref{le:path} hold here, since the only place in that proof where we used the assumption $L^C \neq \emptyset$ was in the case when the tree decomposition had $join$ nodes.
\qed
\end{proof}


The definition of a valid quintuple is the same as for treewidth, but of course in this case we will not use $join$ nodes. Next result is a restriction of Lemma~\ref{le:tvalidq}:
%

\begin{lemma}\label{le:pvalidq}
A quintuple $(\tau_-,L^C,X^C,R^C,\tau_+)$ is valid if $(L^C, X^C,R^C)$ is a valid partition of $C$ and the following two conditions hold:
\begin{itemize}
\item $\tau_-$ is of type 
\begin{itemize}
\item $introduce(u)$ for some $u \in X^C$ and $u \cap L^C = \emptyset$, or
\item $forget(u)$ for some $u \in L^C$
\end{itemize}
\item $\tau_+$ is of type
\begin{itemize}
\item $introduce(v)$ with $v \in R^C$, or
\item $forget(v)$ with $v \in X^C$ and $N(v) \cap R^C = \emptyset$. 
\end{itemize}
\end{itemize}
\end{lemma}

Given a valid quintuple $Q$, sets $XTR^S(Q)$, $XL^S(Q)$ and $XR^S(Q)$ are defined like in Notation~\ref{no:tsets}.
%
%
%
Lemma~\ref{le:timinimax} also holds in this case.

%
%

Unlike in the case of treewidth, the situation is more complicated with the vertices of $S$ whose neighborhood is contained in $X^C$. 
If $N(x) \subseteq X^C_{imin-1} \subsetneq X^C$, 
then we should rather put vertex $x$ in bag $imin -1$ or before (as we shall
see, this will not increase the width of the decomposition). Symmetrically, if 
$N(x) \subseteq X^C_{imax+1} \subsetneq X^C$ then $x$ should be put in some bag
after $imax+1$. If none of these holds, we can create a bag in the subpath $[imin,
imax]$ containing only $X^C \cup XTR^S(Q) \cup \{x\}$. 

\begin{notation}\label{no:XF}
Let $Q = (\tau_-,L^C,X^C,R^C,\tau_+)$ be a valid quintuple. Let $X^C_-$ (resp. $X^C_+$) correspond to $X^C$ before operation $\tau_-$ (resp. to $X^C$ after operation $\tau_+$).

Let $XF^S(Q)$ be the set  of vertices $x\in S$ such that
\begin{itemize}
\item  $N(x) \subseteq X^C$, and
\item if $X^C_- \subsetneq X_C$ then $N(x) \not\subseteq X^C_-$, and 
\item if $X^C_+ \subsetneq X_C$ then $N(x) \not\subseteq X^C_+$.
\end{itemize}
 
Let $\epsilon(Q)$ be set to $1$ if $XF^S(Q)$ is not empty, set to $0$ otherwise.
\end{notation}


\begin{lemma}\label{le:epsilon}
Let ${\cal P}$ be a path decomposition respecting a quintuple $Q = (\tau_-,L^C,X^C,R^C,\tau_+)$. Then ${\cal P}$ has a bag of
size at least $|X^C| + |XTR^S(Q)| + \epsilon(Q)$.
\end{lemma}
\begin{proof}
If $\epsilon(Q) = 0$ then the result comes directly from Lemma~\ref{le:timinimax}. Assume that $\epsilon(Q) = 1$. We distinguish three cases, depending on operations $\tau_-$ and $\tau_+$. 

\paragraph{Case 1: $\tau_- = introduce(u)$ and $\tau_+ = forget(v)$.}  Let $x$ be a vertex of $XF^S(Q)$. Note that $x$ is adjacent to $u$ in $G$ (otherwise $N(x) \subseteq X^C_- \subsetneq X^C$, contradicting $x \in XF^S(Q)$) and symmetrically $x$ is adjacent to $v$ (otherwise $N(x) \subseteq X^C_+ \subsetneq X^C$). Therefore $x$ must be in some bag above $imin$ and in some bag below $imax$. Hence there is some bag in the $[imin, imax]$ subpath containing $x$. By Lemma~\ref{le:timinimax}, that bag is of size at least $|X^C| + |XTR^S(Q)| + 1$.

\paragraph{Case 2: $\tau_+ = introduce(v)$.} We claim that the bag $X_{imax+1}$ contains 
$XTR^S(Q)$. Let $Q_+$ be the valid quintuple corresponding to bag $X_{imax+1}$. The quintuple $Q_+$ is of the type $(introduce(v), L^C,$ $X^C \cup \{v\}, R^C \setminus \{v\}, \tau_{++})$. 
Observe that each vertex $x$ of $XTR^S(Q)$ belongs to $XTR^S(Q_+)$ (if $x$ has a neighbor in $R^C \setminus \{v\}$) or to $XL^S(Q_+)$ (if the only neighbor of $x$ in $R^C$ is $v$). 
Applying Lemma~\ref{le:timinimax} to the valid quintuple $Q_+$ we deduce that $X_{imax+1}$ contains $X^C \cup \{v\} \cup XTR^S(Q)$. Recall that $XTR^S(Q) \subseteq XTR^S(Q_+) \cup XL^S(Q_+)$. Consequently, the size of the bag $X_{imax+1}$ is at least $|X^C| + |XTR^S(Q)| + 1$.

\paragraph{Case 3: $\tau_- = forget(u)$.} This case is perfectly symmetric to the previous one. Let $Q_-$ be the valid quintuple $(\tau_{--}, L^C \setminus \{u\}, X^C \cup\{u\}, R^C, forget(u))$ corresponding to bag $X_{imin-1}$. Then by symmetric arguments $XTR^S(Q) \subseteq XTR^S(Q_-) \cup XR^S(Q_-)$ and hence by Lemma~\ref{le:timinimax} the set $X^C \cup\{u\} \cup XTR^S(Q)$ is contained in bag $X_{imin-1}$.
\qed
\end{proof}

The definition of local pathwidth is slightly different from the local treewidth.
\begin{notation}\label{no:bagsize}
Let $Q = (\tau_-,L^C,X^C,R^C,\tau_+)$ be a valid quintuple. We define the \emph{local pathwidth} of $Q$ as
$$locpw(Q) =|X^C| + |XTR^S| + \max\{|XL^S|, |XR^S|, \epsilon(Q)\}-1.$$
\end{notation}

By Lemmata~\ref{le:timinimax} (which also holds for the pathwidth case) and~\ref{le:epsilon} we deduce:
\begin{corollary}\label{co:bagsize}
Any nice path decomposition of $G$ respecting a valid quintuple $Q$ is of width at least $locpw(Q)$.
\end{corollary}


\begin{notation}\label{no:ppw}
Given a valid quintuple $Q= (\tau_{-},L^C,X^C,R^C,\tau_{+})$, we define the quantity $ppw(Q)$ (like \emph{partial pathwidth}) as follows.

\begin{itemize}
\item If $Q = (introduce(u),\emptyset, \{u\}, C \setminus \{u\}, \tau_+)$ then
$$ppw(Q) = locpw(Q),$$

\item else if $\tau_{-} = introduce(u)$,
 
 $$ppw(Q) = \max \left\{locpw(\tau_{-},L^C,X^C,R^C,\tau_{+}), \min_{\text{valid quintuple~} Q_-} ppw(Q_-)\right\}$$

where the minimum is taken over all valid quintuples $Q_-$ of type $(\tau_{--},L^C,X^C \sm \{u\},R^C \cup \{u\}),$ $introduce(u))$.

\item else ($\tau_{-} = forget(u)$),

 $$ppw(Q) = \max \left\{locpw(\tau_{-},L^C,X^C,R^C,\tau_{+}), \min_{\text{valid quintuple~} Q_-} ppw(Q_-)\right\}$$ 

where the minimum is taken over all valid quintuples $Q_-$ of type $(\tau_{--},L^C\setminus\{u\},X^C \cup \{u\},R^C),$ $forget(u))$.
\end{itemize}
\end{notation}

\begin{lemma}\label{le:geppw}
Any path decomposition respecting a valid quintuple $Q$ is of width at least $ppw(Q)$.
\end{lemma}
\begin{proof}
Recall that the three-partitions $(L^C,X^C,R^C)$ of $C$ can be naturally ordered by the precedence relation (Definition~\ref{de:prec}). Here we restrict this relation to path decomposition, thus the successors of  $(L^C,X^C,R^C)$  are obtained by moving a vertex from $R^C$ to $X^C$ (as happens by operation $introduce$) or from $X^C$ to $L^C$ (as for $forget$).  We prove the lemma for every valid quintuple $(\tau_-,L^C,X^C,R^C,\tau_+)$, by induction (according to this order) on $(L^C,X^C,R^C)$. 
The minimal elements are of type $Q = (introduce(u),\emptyset, \{u\}, C \setminus \{u\}, \tau_+)$, and the result follows from Corollary~\ref{co:bagsize} and the definition of $ppw$ for this case. 

Now let ${\cal P}$ be a path decomposition respecting $(\tau_-,L^C,X^C,R^C,\tau_+)$, and let as before denote by $[imin,imax]$ the subpath of bags whose trace on $C$ is $(L^C,X^C,R^C)$. 
In particular, $imin \geq 2$. Let $(L^C_-,X^C_-,R^C_-)$ be the trace of node $imin-1$ on $C$ and let $Q_- = (\tau_{--},L^C_-,X^C_-,R^C_-,\tau_-)$ be the valid quintuple respected by $\cal P$. Note that $Q_-$ is one of the candidates for valid quintuples preceding $Q$, used in the $\min$ of the formula for $ppw(Q)$ (see Notation~\ref{no:ppw}). By induction hypothesis, the width of $\cal P$ is at least $ppw(Q_-)$. We conclude by Corollary~\ref{co:bagsize}.
\qed
\end{proof}

\begin{theorem}\label{th:pw}
The pathwidth of $G$ is
	$$pw(G) = \min_{Q_{\text{last}}} ppw(Q_{\text{last}})$$
over all valid quintuples $Q_{\text{last}}$ of the form $(\tau_-,C \setminus \{univ\}, \{univ\}, \emptyset, forget(univ))$.
\end{theorem}
\begin{proof}
First note that $pw(G) \geq \min_{Q_{\text{last}}} ppw(Q_{\text{last}})$. Indeed, there exists a path decomposition $\cal P$ of minimum  width that has a root whose trace is $(C \setminus \{univ\}, \{univ\}, \emptyset)$. The inequality follows from Lemma~\ref{le:geppw}.

Conversely, let $Q_{\text{last}}=(\tau_-,C \setminus \{univ\}, \{univ\}, \emptyset, forget(univ))$ be a valid quintuple of minimum $ppw$, among all quintuples of this type; denote by $t$ this minimum value. By definition of $ppw$, there exists a sequence $Q_1, Q_2, \dots, Q_p$ of valid quintuples, the last being $Q_{\text{last}} = (\tau_-,C \setminus \{univ\}, \{univ\}, \emptyset, forget(univ))$, the first being $Q_1 = (introduce(u),\emptyset, \{u\}, C \setminus \{u\}, \tau_{+1})$, such that we go from $Q_j$ to $Q_{j+1}$ by an $introduce$ or a $forget$ operation. Moreover, by definition of $ppw$, all these $Q_j$s have a local pathwidth $locpw(Q_j)$ at most $t$.
We shall now construct a path decomposition of $G$ of width $t$. 

Let $I=(X^C_1, \dots, X^C_p)$ be a linear arrangement of bags, where $Q_j =( \tau_{-j}, L^C_j, X^C_j, R^C_j, \tau_{+j})$. By construction of the sequence $Q_j$, we have that $I$ is a nice path-decomposition of $G[C]$. It remains to place the vertices of $S$ in order to transform it into a path decomposition of $G$. 
Firstly, replace in $I$ each bag $X^C_j$ by a sequence of $|XF^S(Q_j)|+2$ copies of it. 
Denote by $[jmin,jmax]$ the new subpath of  $|XF^S(Q_j)|+2$ nodes that map to bags equal to $X^C_j$. Now
\begin{itemize}
\item add $XTR^S(Q_j)$ to all bags in the subpath $[jmin,jmax]$;
\item add $XL^S(Q_j)$ to bag numbered $jmin$;
\item add $XR^S(Q_j)$ to bag numbered $jmax$;
\item add each vertex of $XF^S(Q_j)$ that does not appear in any $XF^S(Q_l)$ with $l<j$ in exactly one of the bags in the $[jmin+1,jmax-1]$ subpath.
\end{itemize}
Let $I'$ be this new linear arrangement of bags, we claim it represents a path decomposition of $G$, of width $t$. Note that the size of each bag in the $[jmin,jmax]$ subpath is at most $locpw(Q_j)$+1, by construction of $I'$ and by Lemmata~\ref{le:timinimax} and~\ref{le:epsilon}.

For any vertex $u \in C$, $u$ is in some bag of $I'$ and the bags containing $u$ form a subpath, because $I$ was a path decomposition of $G[C]$. For the same reason, every edge of $G$ with both endpoints in $C$ has its ends in a same bag of $I'$. It remains to show that similar conditions hold for vertices of $S$ and edges with an endpoint in $S$ and one in $C$. 

Let $x$ be a vertex of $S$. Let $HC$ be the interval graph induced by $I$ on vertex set $C$, i.e. two vertices are adjacent in $HC$ if and only if there is a bag of $I$ containing both of them. We shall distinguish two cases, depending on whether the neighborhood $N_G(x)$ of $x$ in $G$ induces a clique on $HC$ or not. 

If $x$ is in some set $XF^S(Q_j)$ for some $1 \leq j \leq p$, then its neighbors in $G$ induce a clique in $HC$. Conversely, if $N_G(x)$ forms a clique in $HC$, then by Helly's property this clique is in some bag $X^C_j$ (see Proposition~\ref{pr:hellyc}). Among all these bags, let $j$ be minimum index such that  
$X^C_j$ is a local minimum (in size), in the sequence $X^C_1, \dots, X^C_p$, among bags containing $N_G(x)$. By definition of sets $XF^S$ (Notation~\ref{no:XF}), we have $x \in XF^S(Q_j)$ and by construction of $I'$, $x$ will be put in exactly one bag of the $[jmin+1, jmax-1]$ subpath. This bag also contains all neighbors of $x$ in $G$. We claim that this is the unique bag of $I'$ containing $x$. Assume the contrary, so $x$ is in some set $XTR^S(Q_i)$ or $XL^S(Q_i)$ or $XR^S(Q_i)$ for some $Q_i$. By definition of these sets (Notation~\ref{no:tsets}), $x$ would have two neighbors $a$ and $b$ in $G$, such that $a$ and $b$ are non adjacent in $HC$ --- a contradiction.

The second case is when $N_G(x)$ is not a clique in $HC$. Among all neighbors of $x$ in $G$, let $a$ be the first vertex that is forgotten in the sequence $Q_j$, and let $b$ be the last vertex introduced in this sequence. Let $X_j$ be the last bag of $I$ containing $a$, and $X_{j'}$ the first containing $b$. In particular $j < j'$, otherwise $N_G(x)$ would induce a clique in $HC$. By Notation~\ref{no:tsets}, we have  $x \in XR^S(Q_j)$,  $x \in XL^S(Q_{j'})$ and  $x \in XTR^S(Q_r)$ for all $j < r < j'$. Moreover, $x$ is not in any other set $XTR^S$, $XL^S$ or $XR^S$ or $XF^S$ corresponding to quintuples of the sequence. Therefore the bags containing $x$ in $I'$ form exactly the subpath $[jmax,j'min]$ and they contain all neighbors of $x$ in $C$. 

We conclude that $I'$ is a path decomposition of $G$ of size $t$, and the equality of the theorem holds. 
\qed
\end{proof}

\begin{theorem}\label{th:pwvc}
The \textsc{Pathwidth} problem can be solved in time $O^*(3^k)$, where $k$ is the size of the minimum vertex cover of the input graph.
\end{theorem}
\begin{proof}
Given an arbitrary graph $G$, we first compute a minimum vertex cover in $O^*(1.28^k)$ (Proposition~\ref{pr:VC}). Then $G$ is transformed into a graph $G'$ by adding a universal vertex $univ$. Let $C$ be the vertex cover of $G'$ obtained by adding $univ$ to the minimum vertex cover of $G$ (hence $|C| = k+1$). The pathwidth of $G'$ is computed as follows.
\begin{enumerate}
\item Compute all valid partitions $(L^C,X^C,R^C)$, by enumerating all three-partitions of $C$ and keeping only the valid ones (Lemma~\ref{le:valid}). This can be done in time $O^*(3^k)$. The number of valid partitions is $O(3^{k+1})$.
\item Compute all valid quintuples using Lemma~\ref{le:pvalidq} and sort them using the precedence relation on the corresponding valid three-partitions. The number of valid quintuples is $O^*(3^k)$, their computation can be performed in time $O^*(3^k)$. The sorting can be done within the same running time, the triples $(L^C,X^C,R^C)$ being sorted by increasing size of $L^C$, and in case of tie-breaks by increasing size of $X^C$ (see Definition~\ref{de:prec} and following remarks).
\item For each valid quintuple $Q$, according to the ordering above, compute by dynamic programming $locpw(Q)$ (Notation~\ref{no:bagsize}) and then $ppw(Q)$ (Notation~\ref{no:ppw}). Since this step is polynomial for each $Q$, the overall running time is still $O^*(3^k)$.
\item Compute the pathwidth of $G'$ using Theorem~\ref{th:pw}, and return $pw(G)=pw(G')-1$. This step takes polynomial running time. 
\end{enumerate}
This achieves the proof of the theorem. The algorithm can also be adapted to return, within the same time bounds, an optimal path decomposition of the input graph.
\qed
\end{proof}

\section{Proofs of Section~\ref{se:tw4k}}\label{se:prtw}
\subsection{Proof of Lemma~\ref{le:tvalidq}}\label{ss:le:tvalidq}
\begin{proof}
Assume that $Q=(\tau_-,L^C,X^C,R^C,\tau_+)$  is a valid quintuple and let $(T,{\cal X})$ be a nice tree-decomposition respecting it. Let $imin$ and $imax$ be the nodes of $T$ like in Lemma~\ref{le:path}. By definition of $introduce$ and $forget$ operations, if $\tau_-=forget(u)$ we must have $u \in L^C$, and if
$\tau_-=introduce(u)$ we must have $u \in X^C$. Also, if $\tau_+=introduce(v)$ it means that for $imax$, vertex $v$ was still in $R^C$, and if $\tau_+=forget(v)$ we must have $v \in X^C$. 
In the latter case, to be able to perform the $forget(v)$ operation, $v$ must have no neighbor in $R^C$, because bag $X_{imax+1}$ separates all vertices from bags below $imax+1$ 
from all vertices not appearing below $imax+1$ (Proposition~\ref{pr:sep} applied to $imax+1$). Symmetrically, if $\tau_- = forget(u)$, then $u$ has no neighbors in $R^C$ because $X^C$ separates $L^C$ from $R^C$, and $u \in L^C$.

If $imin$ is $join(X^C;L1^C,L2^C)$, then $(L1^C,X^C,R^C \cup L2^C)$ and $(L2^C,X^C,R^C \cup L1^C)$ are the traces of the two sons $j$ and $k$ of $imin$. Both three-partitions are thus valid. Moreover $L1^C$ (resp. $L2^C$) is not empty, otherwise $j$ (reps. $k$) would have the same trace as $imin$ -- a contradiction. If $imax+1$ is an $introduce$ of $forget$ node, the condition on $\tau_+$ follows by the same arguments as in Lemma~\ref{le:pvalidq}. If $\tau_+$ is a $join$ node, then it is of type $join(X^C; L^C,LR^C)$. Note that $LR^C \neq \emptyset$ (otherwise $imax+1$ would leave the same trace as $imax$). The three-partitions $(
L^C \cup LR^C, X^C, R^C \setminus LR^C)$ and  $(LR^C, X^C, R^C\setminus LR^C \cup L^C)$ are the respective traces of $imax+1$ and of the sibling of $imax$, so they are valid.

Conversely, let $(L^C,X^C,R^C)$, $\tau_-$ and $\tau_+$ satisfy the conditions of the lemma, we construct a nice tree decomposition respecting $Q$. We start with a unique node $i$ with bag $X^C$. 
If $\tau_-$ is an $introduce$ or $forget$ node, let $(L^C_-,X^C_-,R^C_-)$  correspond to $(L^C,X^C,R^C)$ before operation $\tau_-$. Add to node $i$ a son $j$ with bag $X^C_-$ and to $j$ a son $j'$ with bag $L^C_- \cup X^C_-$. If $\tau_-= join(X^C; L1^C, L2^C)$ we add to $i$ two sons $j$ and $k$ with bags $X^C$ and add to $j$ (resp. $k$) a son $j'$ with bag $X^C \cup L1^C$ (resp. a son $k'$ with bag $X^C \cup L2^C$). If $\tau_+$ is an $introduce$ or $forget$ node, let $(L^C_+,X^C_+,R^C_+)$  correspond to $(L^C,X^C,R^C)$ after operation $\tau_+$. Add to $i$ a father 
$l$ with bag $X^C_+$ and to $l$ a father $l'$ with bag $X^C_+ \cup R^C_+$. If $\tau_+ = join(X^C; L^C,LR^C)$ then add to $i$ a father $l$ and a sibling $r$ with bag $X^C$. We add a father $l'$ of $l$, with bag $X^C \cup R^C \setminus LR^C$, and to $r$ a son $r'$ with bag $X^C \cup LR^C$. Observe that at this stage we have a tree-decomposition of $G[C]$. Add $S$ to each bag of the tree decomposition, we obtain a tree decomposition of $G$. By refining it into a nice one (Proposition~\ref{pr:nice}), we obtain a nice tree decomposition of $G$. The original node $i$ plays both the role of $imin$ and $imax$ in the new tree decomposition, and $Q$ is the valid quintuple for it.
\qed
\end{proof} 

\subsection{Proof of Theorem~\ref{th:tw}}\label{ss:th:tw}

\begin{proof}
For the sake of readability, we briefly recall the construction of $(T, {\cal X})$.

We have started from the tree decomposition $(T^C,{\cal X}^C)$ of $G[C]$, which is a nice tree decomposition of $G[C]$, except for the leaf nodes whose bags are not necessarily of size one.
The tree $T$ is obtained from $T^C$ by replacing each node $i$ with a path of three nodes, denoted $imin, imid$ and $imax$ (from the bottom towards the top). 
f
Initially, we associate to the three nodes $imin, imid, imax$ the same bag $X^C_i$. Now, for each $i$,
\begin{enumerate}
\item add $XTR^S(Q_i)$ to all bags in the subpath $[imin,imax]$ of $T$;
\item add $XL^S(Q_i)$ to bag number $imin$;
\item add $XR^S(Q_i)$ to bag number $imax$;
\item For each vertex $x \in XF^S(Q_i)$, which has not yet been added to some bag of $T$, create a new node of $T$ adjacent only to $imid$ and associate to this node the bag $N[x]$. These nodes are called \emph{pending nodes}.
\end{enumerate}

It remains show that $(T, {\cal X})$ is a tree decomposition of $G$, more precisely that, for every vertex $x \in S$, the bags containing $x$ form a connected subtree of $T$, and the edges of $G$ incident to $x$ are covered by some bag.

Let $HC$ be the graph on vertex set $C$, where two vertices are adjacent if and only if they belong to a same bag of $(T^C,{\cal X}^C)$ (see also Proposition~\ref{pr:hellyc}).
We distinguish two types of vertices of $S$: vertices $x$ whose neighborhood $N_G(x)$ induces a clique in $HC$, and vertices whose neighborhood in $G$ does not induce a clique in $HC$.

\noindent
\textbf{Claim 1.} A vertex $x \in S$ is of the first type, i.e. $N_G(x)$ induces a clique in $HC$, if and only if $x \in XF^S(Q_i)$ for one of the selected quintuples $Q_i$.

By definition of $XF^S(Q_i)$, if $x$ belongs to this set then $N_G(x) \subseteq X^C_i$ and thus $N_G(x)$ forms a clique in $HC$. Conversely, if $N_G(x)$ is a clique in $HC$, by Proposition~\ref{pr:hellyc} there is some bag $X^C_i$ containing $N_G(x)$. Therefore $x \in XF^S(Q_i)$.

\noindent
\textbf{Claim 2.} A vertex $x \in S$ is of the second type, i.e. $N_G(x)$ does not induce a clique in $HC$, if and only if there is a selected quintuple $Q_i$ such that $x \in XTR^S(Q_i)$, $x \in XR^S(Q_i)$ or $x \in XL^S(Q_i)$.

Assume there is a selected quintuple such that $x \in XTR^S(Q_i)$. Then $x$ has two neighbors $u,v$ in $G$ such that $u \in L^C_i$ and $v \in R^C_i$. Note that $(T^C, {\cal X}^C)$ is a tree decomposition of $HC$, so by Proposition~\ref{pr:sep}, $X^C_i$ separates, in graph $HC$, the vertices $u \in L^C_i$ and $v \in R^C_i$. In particular, $N_G(x)$ is not a clique in $HC$. In the case when $x \in XL^S(Q_i)$, we have two possibilities. If $\tau-_i = join(X^C_i; L1^C_i, L2^C_i)$, then it means that $x$ has, in graph $G$, two neighbors $u \in L1^C_i$ and $v \in L2^C_i$. Again $X^C_i$ separates these two vertices in $HC$ thus $N_G(x)$ is not a clique in $HC$. If $\tau-_i = introduce(u)$, then $u$ is a neighbor of $x$ in $G$, and $x$ also has another neighbor $v \in L^C_i$. Let $j$ be the unique son of $i$ in $T^C$, in particular $X^C_j$ does not contain $u$. Note that $x \in XTR^S(Q_j)$ because $v \in L^C_j$ and $u \in R^C_j$, so again $N_G(x)$ is not a clique in $HC$. The case $x \in XL^S(Q_i)$ is symmetrical to 
this last case.

Conversely, assume that $N_G(x)$ is not a clique in $HC$ and let $u,v \in N_G(x)$ be two vertices, non-adjacent in $HC$. Thus, in the tree decomposition $(T^C, {\cal X}^C)$ of $G[C]$, there is a bag $i$ separating, in $T^C$, the bags containing $u$ from the bags containing $v$. Let $Q_i$ be the corresponding valid quintuple. Either $u \in L^C_i$ and $v \in R^C_i$ (or symmetrical), in which case $x$  belongs to $XTR^S(Q_i)$, or $i$ is a $join(X^C_i; L1^C_i, L2^C_i)$ and $u,v$ belong to different parts $L1^C_i$ and $L2^C_i$. In the latter case, if $x$ has neighbors in $R^C_i$ then $x$ belongs to $XTR^S(Q_i)$, otherwise it belongs to $XL^C(Q_i)$. We conclude that $x \in XTR^S(Q_i)$ or $x \in XL^S(Q_i)$.

\medskip
The two claims ensure that every vertex $x \in S$ appears in some bags of the tree decomposition $(T, {\cal X})$. Moreover, if $x$ is in the second case, then $x$ appears in a unique, pending bag. Hence the bags containing $x$ form a (trivial) subtree in $T$, and since this unique bag contains $N_G[x]$ it covers all edges incident to $x$ in $G$.

It remains to prove that for every vertex $x$ of the second type, the bags of $(T, {\cal X})$ form a subtree of $T$ and cover all edges incident to $x$ in $G$.
Assume $x$ appears in the bags of two nodes $a$ and $b$ of $T$, and let $c$ be a node on the $a,b$-path of the decomposition tree. We must prove that $x$ is in the bag of $c$. Clearly, $a$ and $b$ are not pending nodes. Suppose first that $a$ and $b$ have $c$ as a common ancestor. Then $c = kmin$ for some node $k$ of $T^C$. Also
$a$ (resp. $b$) come from some node $i$ (resp. $j$) of $T^C$. Let $Q_i =  (\tau_{-i}, L^C_i, X^C_i, R^C_i, \tau_{+i})$, $Q_j =  (\tau_{-j}, L^C_j, X^C_j, R^C_j, \tau_{+j})$ and $Q_k =  (\tau_{-k}=join(X^C_k;L1^C_k, L2^C_k), L^C_k, X^C_k, R^C_k, \tau_{+k})$ be the valid quintuples corresponding to nodes $i$, $j$ and $k$. Since $x$ has been added to node $a$, $x$ is in $XTR^S(Q_i)$, $XL^S(Q_i)$ or $XR^S(Q_i)$. In the two first cases, it means that $x$ has a neighbor in $L^C_i$. Note that $L^C_i \subseteq L1^C_k$. In the third case, it means that $\tau_{+i}$ is of type $forget(u)$ and $u$ is a neighbor of $x$. Again, $u \in L1^C_k$. Therefore $x$ has a neighbor in $L1^C_k$. For symmetrical reasons, $x$ has a neighbor in $L2^C_k$. Hence, $x$ is in $XTR^S(Q_k)$ or
$XL^S(Q_k)$, and consequently in the bag of node $c=kmin$ of the tree decomposition $(T, {\cal X})$. 

We consider now w.l.o.g. that $b$ is an ancestor of $a$. With the same notations as above,
$x$ has a neighbor in $L^C_i$, or $a=imax$, $\tau_{i+} = forget(u)$ and $u$ is a neighbor of $x$. Since $c$ is an ancestor of $a$, we have $L^C_i \subseteq L^C_k$, hence $x$ has a neighbor in $L^C_k$. Since $x$ is in the bag of node $b$, we have one of the following
\begin{itemize}
\item $x \in XTR^S(Q_j)$, thus $x$ has a neighbor in $R^C_j$,
\item $b = jmin$ and $x \in XL^S(Q_j)$, 
thus $\tau_{-j} = join(L^C_i;L1^C_j,L2^C_j)$ and $x$ has neighbors in both sets $L1^C_j,L2^C_j$.
\end{itemize}
Note that the case $x \in XR^S(Q_j)$ is not possible, because $x$ has a neighbor in $L^C_i \subseteq L^C_j$. Note that $R^C_j \subseteq R^C_k$ and, if we are in the second case, $L2^C_i \subseteq R^C_k$ (assuming that $k$ is a left descendent of $j$). Therefore, we must have $x \in XTR^S(Q_k)$, and $x$ is in the bag of node $c$.

We have shown that for each $x \in S$ of the second type, the bags of $(T, {\cal X})$ containing it form a connected subtree of $T$. Let us prove that for each edge incident to $x$ in $G$, both ends are in some bag of ${\cal X}$. We first study the placement of $N_G(x)$ in the bags of $(T^C, {\cal X}^C)$. Let $r(x)$ be the lowest node of $T^C$ such that every vertex of $N_G(x)$ appears in the subtree of $T^C$ rooted in $r(x)$. Equivalently, $r(x)$ is the lowest node of $T^C$ such that the corresponding quintuple $Q_{r(x)} =(\tau_{-r(x)}, L^C_{r(x)}, X^C_{r(x)}, R^C_{r(x)}, \tau_{+r(x)})$ satisfies $R^C_{r(x)} \cap N_G(x) = \emptyset$. Note that $r(x)$ is unique. We prove that $x \in XL^S(Q_{r(x)})$. If $r(x)$ has a unique son $j$, then $r(x)$ must be of type $introduce(u)$ for some $u \in N_G(x)$ (otherwise we could replace $r(x)$ by $j$, contradicting its definition). Also $x$ must have some neighbor in $L^C_{r(x)}$ (otherwise $N_G(x) \subseteq   X^C_{r(x)}$, contradicting the fact that $x$ is of the second 
type). Hence $x \in  XL^S(Q_{r(x)})$ by Notation~\ref{no:tsets}. Assume now that $r(x)$ is a $join$ node, so $\tau_{-r(x)} = join(X^C_{r(x)}; L1^C_{r(x)}, L2^C_{r(x)})$. Then $N_G(x)$ intersects both $L1^C_{r(x)}$ and $L2^C_{r(x)}$, otherwise we could replace $r(x)$ by one of its sons. Again,  $x \in  XL^S(Q_{r(x)})$ by Notation~\ref{no:tsets} applied to $join$ nodes. Now the fact that $x \in  XL^S(Q_{r(x)})$ implies that $x$ belongs to bag $r(x)min$ in the decomposition $(T, {\cal X})$.

Let now $\{l_1(x), \dots, l_p(x)\}$ be the set of lowest nodes $l$ of $T^C$ such that the quintuple $Q_{l}= (\tau_{-l}, L^l, X^l, R^l, \tau_{+l})$  satisfies $\tau_{+l} = forget(u)$, for some $u \in N_G(x)$. We prove that $x \in XR^S(Q_l)$. Indeed $x$ has no neighbors in $L^l$ (otherwise $l$ would not have been a lowest node with the required property) and $x$ must have some neighbor in $R^l$ (otherwise $N_G(x) \subseteq X^l$, contradicting the fact that $x$ is of the second type). Hence, by Notation~\ref{no:tsets}, $x$ belongs to $XR^S(Q_l)$, and by our construction it also belongs to the bag of node $lmax$ in $(T, {\cal X})$. 

Let $T(x)$ be the minimal subtree of $T$ spanning the nodes $r(x)min, l_1(x)max, \dots, l_p(x)max$. We have proven above that the bags of $(T, {\cal X})$ containing $x$ form a subtree of $T$. Therefore $x$ belongs to all these all bags of $(T, {\cal X})$, of the subtree $T(x)$. It remains to argue that these bags include all vertices of $N_G(x)$. Let $u \in N_G(x)$ (we may assume that $u \neq univ$, because the property is trivial for vertex $univ$). 
In the tree decomposition $(T^C,{\cal X}^C)$, clearly $u$ appears in some bag of the subtree rooted in $r(x)$ (because $u \not\in R^C_{r(x)}$). Let $j$ be the node with operation $forget((u)$. There must be some node $l_q(x)$ in the subtree of $T^C$ rooted in $j$. Also note that $l_q(x)$ is in the subtree of $T^C$ rooted in $r(x)$. Hence $u$ is in some bag of $(T^C, {\cal X}^C)$, on the path from $l_q(x)$ to $r(x)$ in $T^C$. Thus $u$ belongs to some bag of  $(T, {\cal X})$, on the path from $l_q(x)max$ to $r(x)min$ of $T$. Hence $u$ belongs to a bag of $T(x)$, which also contains vertex $x$. We deduce that all edges incident to $x$ in $G$ are covered by some bag of $(T, {\cal X})$.

This concludes the proof of the theorem. 
\qed
\end{proof}

\section{An $O^*(3^k)$ algorithm for \textsc{treewidth}}\label{se:tw3k}

Let us have a deeper look at the partial treewidth parameter $ptw(Q)$. Let $(T^C, {\mathcal X}^C)$ be a nice tree decomposition of $G[C]$ (again the root only contains the universal vertex). The trace of each node $i$ is some valid triple $(L_i^C, X_i^C, R_i^C)$. We \emph{trim} this tree decomposition, removing all nodes $i$ such that $L_i^C = \emptyset$, except if the parent of $i$ is a $forget$ node. We call such a tree decomposition a \emph{trimmed} tree decomposition of $G[C]$. At this stage, to each node $i$ of the trimmed tree decomposition corresponds a valid quintuple $Q_i$ (for the leaf nodes, the corresponding valid quintuples are degenerated). Let $Q_i$ be a valid quintuple corresponding to some node $i$ of a trimmed tree decomposition $(T^C, {\mathcal X}^C)$. We call \emph{$Q_i$-rooted subtree decomposition} the part of $(T^C, {\mathcal X}^C)$ induced by the subtree rooted in $i$. 

Note that an equivalent definition of $ptw(Q)$ is the following: $ptw(Q)$ is the minimum value $t$ such that there exists a $Q$-rooted subtree decomposition whose nodes correspond to valid quintuples $Q'$ of local treewidth $loctw(Q')$ at most $t$. Let us add a constraint on the structure of the $Q$-rooted subtree decomposition, more precisely on the ``depth'' of this subtree in terms of  $join$ nodes.

\begin{notation}
Let $d \in \{0,\dots, k\}$ and let $Q$ be a valid quintuple. We denote by $ptw(d,Q)$ the minimum integer $t$ such that there exists a $Q$-rooted subtree decomposition satisfying the following conditions:
\begin{itemize}
\item For any valid quintuple $Q'$ corresponding to a node of the $Q$-rooted subtree decomposition, $loctw(Q') \leq t$.
\item On any path from the root of the subtree to a leaf, there are at most $d$ nodes of type $join$. 
\end{itemize}
\end{notation}

In particular, $ptw(0,Q)$ is the minimum value that can be obtained using tree decomposition without any $join$ node in the subtree of $Q$. Clearly, $ptw(Q) = ptw(k,Q)$, since any trimmed tree decomposition contains at most $k$ (actually even at most $k-1$) join nodes from the root to a leaf. 

The principle of our algorithm is to compute, for each $d$ from $0$ to $k$, the values $ptw(d,Q)$ using the previously computed values. 
In order to gain in time complexity from $O^*(4^k)$ to $O^*(3^k)$, we can not afford to consider each quintuple of the type $$(join(X^C ; L1^C, L2^C), L^C, X^C, R^C, \tau_+).$$ (Note that, as before, if the quintuple $Q$ has a $\tau_+$ is of type $join$, the parameters of the join are irrelevant for the value of $ptw(d,Q)$). Instead of computing $$ptw(d, (join(X^C ; L1^C, L2^C), L^C, X^C, R^C, \tau_+))$$ we shall simply store a unique value $$ptw(d, (join, L^C, X^C, R^C, \tau_+)),$$ which is defined as
	$$\min ptw(d, (join(X^C ; L1^C, L2^C), L^C, X^C, R^C, \tau_+)),$$
over all non-trivial partitions $(L1^C, L2^C)$ of $L^C$ (moreover there must exist valid quintuples of the form $(\tau1_-, L1^C, X^C, R^C \cup L2^C, join)$ and $(\tau2_-, L2^C, X^C, R^C \cup L1^C, join)$).

All these values $ptw(d,(join,\dots))$ will be computed from values $ptw(d-1,\dots)$, by subset convolution, using a specific algorithm \textsc{JoinPtw} (Algorithm~\ref{al:JoinPtw}). Assume that one call of Algorithm  \textsc{JoinPtw} runs in time $O^*(3^k)$. Then Algorithm \textsc{Ptw} (Algorithm~\ref{al:Ptw}) correctly computes all values $ptw(d,\dots)$, in  $O^*(3^k)$ time. Eventually, using $d=k$, we obtain the treewidth of graph $G$ within the same time complexity. 

\begin{algorithm}
\KwIn{$d$ and, if $d \geq 1$, all values $ptw(d-1,\dots)$}
\KwOut{all values $ptw(d,\dots)$}
\eIf{$d=0$}{
	set all values $ptw(0,(join,\dots))$ to $\infty$\;
}
{
	use Algorithm~\textsc{JoinPtw} to compute all values $ptw(d,(join,\dots))$ from $ptw(d-1,\dots)$\;
}

\ForEach{valid quintuple $Q=(\tau_-,L^C,X^C,R^C,\tau+)$ s.t. $\tau_-\neq join$, by increasing order on $(L^C,X^C,R^C)$}
{
         \If{$Q$ is degenerate}
         {$ptw(d,Q) = loctw(Q)$\;}
	\If{$\tau_- = introduce(u)$}
	 {$ptw(d,Q) = \max(loctw(Q), \min_{\text{valid quintuple~} Q_-} ptw(Q_-))$ over all valid quintuples $Q_-$ of type $(\tau_{--},L^C,X^C \sm \{u\},R^C \cup \{u\}),introduce(u))$\;
	}
	\If{$\tau_- = forget(u)$}
	{$ptw(d,Q) = \max(loctw(Q), \min_{\text{valid quintuple~} Q_-} ptw(Q_-))$ over all valid quintuples $Q_-$ of type $(\tau_{--},L^C\setminus\{u\},X^C \cup \{u\},R^C),forget(u))$\;
	}
}
\caption{\textsc{Ptw} computes $ptw(d,\dots)$}
\label{al:Ptw}
\end{algorithm}

It remains to describe Algorithm \textsc{JoinPtw}. The key algorithmic tool is the following result on subset convolution~\cite{BHKK07}:

\begin{theorem}[~\cite{BHKK07}]\label{th:convolution}
Let $U$ be a universe of size $s$ and $f,g$ be two functions from $2^U$ to $\mathbb{N}$. Assume these functions are stored as tables. The function $f \star g : 2^U \rightarrow \mathbb{N}$ defined as
$$f \star g (X) = \sum_{Y \subseteq X} f(Y) \cdot g(X \setminus Y)$$
is called the \emph{subset convolution} of $f$ and $g$. 

If the maximum values of $f, g$ and $f \star g$ are at most $2^s$, then $f \star g$ can be computed in time $O^*(2^s)$. 
\end{theorem}

Recall that we aim to compute all quantities of the form 
$$ptw(d, (join, L^C, X^C, R^C, \tau_+))$$
in $O^*(3^k)$ time. For this purpose, We shall fix a set $X^C \subseteq C$ and use Theorem~\ref{th:convolution} to compute the corresponding quantities in time $O^*(2^{k - |X^C|})$, using only 
values of the type $ptw(d-1,\dots)$.

The family of functions defined above will be used in the subset convolution:
\begin{notation}
Let $X^C \subseteq C$. Let $d,l$ be two integers $0 \leq d,l \leq k$ and let $t$ be a positive integer. We define the function 
$$ptw\_atmost\_t_{d,X^C}^{t,l} : 2^{C \setminus X^C} \rightarrow \{0,1\}$$
as follows:
$$ptw\_atmost\_t_{d,X^C}^{t,l}(L^C) = \left\{
\begin{array}{ l l }
1 & \mbox{if there is a valid quintuple~$Q=(\tau_{-},L^C,X^C,R^C,join)$}\\
    & \mbox{~s.t. $ptw(d,Q)\leq t$ and}\\
    & \mbox{ $|\{x \in S \mid N(x) \cap L^C \neq \emptyset \mbox{~and~} N(x) \cap R^C = \emptyset \}| = l$}\\
0 & \mbox{otherwise} 
\end{array}\right.$$ 
\end{notation}

\begin{lemma}\label{le:joinptw}
Let $\tQ=(join,L^C,X^C,R^C,\tau_+)$ correspond to a simplified valid triple (without the parameters of the $join$ operation). Sets $XTR^S(\tQ)$, $XR^S(\tQ)$ and the value $\epsilon(\tQ)$ are defined like in Notation~\ref{no:tsets} and Lemma~\ref{le:tepsilon}\footnote{These values are uniquely defined, but observe that we cannot use the same notation to define the set $XL^S(\tQ)$}. Let $l$ be the size of the set $\{x \in S \mid N(x) \cap L^C \neq \emptyset \mbox{~and~} N(x) \cap R^C = \emptyset\}$. 

Then $ptw(d,join,L^C,X^C,R^C,\tau_+)$ is the minimum integer $t$ such that there exists integers $l_1,l_2$, $0 \leq l_1, l_2 \leq k$ and a two-partition $(L1^C, L2^C)$ of $L^C$ satisfying:
\begin{enumerate}
\item $|X^C|+\max\left\{|XTR^S(\tQ)|+l-l_1-l_2, |XTR^S(\tQ)|+|XR^S(\tQ)|,\epsilon(\tQ)\right\} \leq t$.
\item $ptw\_atmost\_t_{d-1,X^C}^{t,l_1}(L1^C) = 1$ and $ptw\_atmost\_t_{d-1,X^C}^{t,l_2}(L2^C) = 1$.
\end{enumerate}
\end{lemma}
\begin{proof}
If $ptw(d,join,L^C,X^C,R^C,\tau_+)\leq t$, there exists a valid quintuple $Q=(join(L1,L2),L^C,$ $X^C,R^C,\tau_+)$ such that $ptw(d,Q)\leq t$. Fix a $Q$-rooted subtree decomposition realizing $ptw(d,Q)$ and let $Q1_-=(\tau1_{--},L1^C,X^C,R^C \cup L2^C, join)$ and $Q2_-=(\tau2_{--},L2^C,X^C,R^C \cup L1^C, join)$ be the valid quintuples corresponding to the sons of the root of the subtree. Let $l_1$ (resp. $l_2$) be the number of vertices of $S$ having neighbors in $L1^C$ but with no neighbors in $R^C \cup L2^C)$ (resp. having neighbors in $L2^C$ but not in $R^C \cup L1^C$). We point out that $|XL^S(Q)| = l - l_1 - l_2$. Since $loctw(Q)\leq t$, this implies the first property. We also observe that $ptw(d-1,Q1_-) \leq t$ and $ptw(d-1,Q2_-) \leq t$, since the $Q1_-$ and the $Q2_-$-rooted subtree decompositions have at most $d-1$ join nodes on any path from their root to a leaf (one less that the $Q$-rooted subtree decomposition). This implies the second property of the theorem.

Conversely, consider $\tQ=(join,L^C,X^C,R^C,\tau_+)$, $t,l,l_1,l_2$, $L1^C$ and $L2^C$ satisfying the conditions of the theorem. We prove that $ptw(d,\tQ) \leq t$. Since $ptw\_atmost\_t_{d-1,X^C}^{t,l_1}(L1^C) = 1$ and $ptw\_atmost\_t_{d-1,X^C}^{t,l_2}(L2^C) = 1$ there exist valid quintuples $Q1_-=(\tau1_{--},L1^C,X^C,R^C \cup L2^C, join)$ and $Q2_-(\tau2_{--},L2^C,X^C,R^C \cup L1^C, join)$ such that
 $ptw(d-1,Q1_-) \leq t$ and $ptw(d-1,Q2_-) \leq t$. Let $Q = (join(L1^C, L2^C),L^C,X^C,R^C,\tau_+)$. Note that $Q$ is a valid quintuple and $|XL^S(Q)| = l-l_1-l_2$. By the first proprety, we have that
 $loctw(Q)\leq t$. Denote by $T_1$ a $Q_1$-rooted subtree decomposition such that each part from the root to a leaf contains at most $d-1$ $join$ nodes and each valid triple of the subtree has a local treewidth at most $t$. Choose a similar $Q_2$-rooted subtree decomposition denoted $T_2$. Let now $T$ be any $Q$-rooted tree decomposition, and let $i_1$ and $i_2$ be the sons of the root. Replace the subtree rooted in $i_1$ (resp. in $i_2$) by $T_1$ (resp $T_2$). In this new $Q$-rooted subtree decomposition, all paths from the root to a leaf contain at most $d$ nodes of type $join$, and all valid quintuples have local treewidth at most $t$. This certifies that $ptw(d,Q) \leq t$ and consequently $ptw(d,\tQ) \leq t$.
\qed
\end{proof}

Algorithm~\ref{al:JoinPtw}, \textsc{JoinPtw}, simply applies Lemma~\ref{le:joinptw} and the observation that the second condition of this lemma can be expressed by a convolution:
\begin{lemma}\label{le:joinptwConv}
The second condition of Lemma~\ref{le:joinptw} is equivalent to:
$$(ptw\_atmost\_t_{d-1,X^C}^{t,l_1} \star ptw\_atmost\_t_{d-1,X^C}^{t,l_2} )(L^C) \geq 1.$$
\end{lemma}
\begin{proof}
``$\Rightarrow$:'' Assume the second condition of Lemma~\ref{le:joinptw} is satisfied. Then there is a two-partition $(L1^C, L2^C)$ of $L^C$ such that $ptw\_atmost\_t_{d-1,X^C}^{t,l_1}(L1^C) = 1$ and  $ptw\_atmost\_t_{d-1,X^C}^{t,l_2}(L2^C) = 1$. Recall that $$(ptw\_atmost\_t_{d-1,X^C}^{t,l_1} \star ptw\_atmost\_t_{d-1,X^C}^{t,l_2} )(L^C)$$
is defined (see \ref{th:convolution}) as 
$$\sum_{Y \subseteq L^C} ptw\_atmost\_t_{d-1,X^C}^{t,l_1}(Y) \cdot ptw\_atmost\_t_{d-1,X^C}^{t,l_2}(L^C \setminus Y).$$
Hence, in the sum over all $Y \subseteq L^C$, when $Y = L1^C$, both factors are equal to one. Therefore the sum is at least $1$.

`$\Leftarrow$:'' Conversely, in the sum above all terms are either 0 or 1. If the sum is at least 1, it means there exists $Y \subseteq L^C$ such that $ptw\_atmost\_t_{d-1,X^C}^{t,l_1}(Y) =1$ and 
$ptw\_atmost\_t_{d-1,X^C}^{t,l_2}(L^C \setminus Y) = 1$. Thus by taking $(L1^C, L2^C) = (Y, L^C \setminus Y)$ we have the two-partition of $L^S$ required by Lemma~\ref{le:joinptw}.
\qed
\end{proof}

\begin{algorithm}
\KwIn{$d$ and, if $d \geq 1$, all values $ptw(d-1,\dots)$}
\KwOut{all values $ptw(d,(join,\dots)$}

\ForEach{$X^C \subseteq C$}
{
	\ForEach{$t,l$ with $0 \leq t \leq n$ and $0 \leq l \leq k$}
	{
		compute all functions $ptw\_atmost\_t_{d-1,X^C}^{t,l}: 2^{C \setminus X^C} \rightarrow \{0,1\}$\;		
	}
	\ForEach{$l_1,l_2$ with $0 \leq l_1,l_2 \leq l$}
	{
		compute, using Theorem~\ref{th:convolution}, all convolutions $ptw\_atmost\_t_{d-1,X^C}^{t,l_1} \star ptw\_atmost\_t_{d-1,X^C}^{t,l_2}$\;
	}
	\ForEach{$L^C$, $R^C$, $\tau+$ such that $\tQ = (join, L^C, X^C, R^C, \tau_+)$ corresponds to valid tuples}
	{
		compute $XTR^S(\tQ), XR^S(\tQ),\epsilon(\tQ)$ using Notation~\ref{no:tsets} and Lemma~\ref{le:tepsilon}\;
		let $l$ be the size of $\{x \in S \mid x$ sees $L^C$ but does not see $C \setminus (L^C \cup X^C)  \}$\;
		find the minimum $t$ such that there exist $l, l_1, l_2$ satisfying (cf. Lemmata~\ref{le:joinptw} and~\ref{le:joinptwConv}) \\
		\ \ \ (1) $|X^C|+\max\left\{|XTR^S(\tQ)|+l-l_1-l_2, |XTR^S(\tQ)|+|XR^S(\tQ)|,\epsilon(\tQ)\right\} \leq t$, and \\
		\ \ \ (2) $(ptw\_atmost\_t_{d-1,X^C}^{t,l_1} \star ptw\_atmost\_t_{d-1,X^C}^{t,l_2} )(L^C) \geq 1$\;
		$ptw(d,(join, L^C, X^C, R^C, \tau_+)) := t$\;
	}
}
\caption{\textsc{JoinPtw} computes $ptw(d,(join,\dots)$}
\label{al:JoinPtw}
\end{algorithm}

Algorithm \textsc{JoinPtw} (Algorithm~\ref{al:JoinPtw}) is correct by Lemma~\ref{le:joinptw}. Its running time is given by the computations of lines 4-5. For each set $X^C \subseteq C$, we compute $O(k^2n)$ subset convolutions as functions from $2^{C \setminus X^C}$ to ${\mathbb N}$. Each of these computations takes, by Theorem~\ref{th:convolution}, $O^*(2^{k-|X^C|})$ time. Recall that 
$$\sum_{|X^C| \subseteq C} 2^{k - |X^C|} = \sum_{i=0}^k \binom{k}{i} 2^{k-i} = 3^k.$$ 
hence the whole running time of the algorithm is of order $O^*(3^k)$.

Altogether, we obtain:
\begin{theorem}\label{th:twvc3k}
The \textsc{Treewidth} problem can be solved in $O^*(3^k)$ time, where $k$ is the size of the minimum vertex cover of the input graph. 
\end{theorem}

\section{\textsc{Pathwidth} parameterized by the vertex cover of the complement}\label{se:pwcomp}

Let $C$ be a vertex cover of $\overline{G}$ and let $S = V \setminus C$ be an independent set in  $\overline{G}$ 
and thus a clique in $G$. Denote $k' = |C|$. Recall that, by Proposition~\ref{pr:nice}, for any tree or path decomposition of $G$, there will be some bag of the decomposition containing $S$.
Boadlender et. al.~\cite{BodlaenderFKKT06}  proved the following. 

\begin{proposition}[\cite{BodlaenderFKKT06}]
There is an algorithm taking as input a graph $G=(V,E)$ and a clique $S$ of $G$ and computing the treewidth of $G$ in $O^*(2^{|V \setminus S|})$ time.
\end{proposition}

Consequently, the treewidth of $G$ can be computed in time $O^*(2^{k'})$. 

The previous lemma can be adapted, with similar techniques, to the pathwidth case.

\begin{lemma}\label{le:wcomp}
Let $G=(V,E)$ be a graph and $W \subset V$ a vertex subset. Let us denote by $pw(G; W)$  the minimum width, over all path decompositions of $G$ having $W$ as an end bag.
 
There is an $O^*(2^{|C|})$ algorithm computing the values $pw(G[N[L]]; N(L))$, for all subsets $L$ of $C$. 
\end{lemma}
\begin{proof}
We use an alternative definition for pathwidth, expressed as a vertex layout problem, see e.g.~\cite{BodlaenderFKKT12}. 
A layout is a total ordering ${\cal L} = (v_1,\dots, v_n)$ of the vertices of $G$. The vertex separation of a layout is 
$$vs({\cal L}) = \max_{1 \leq i \leq n} |N(\{v_i, v_{i+1}, \dots, v_n\})|.$$
 
Consider a path decomposition $(P,{\cal X})$ of $G$. We can associate a layout ${\cal L}_{(P,{\cal X})}$ as follows~: for any two vertices $u$ and $v$, if the lowest bag containing $u$ is strictly before the lowest bag containing $v$, then $u$ must appear before $v$ in the layout. It is known~\cite{Kinnersley92} that $vs({\cal L}_{(P,{\cal X})}) \leq width((P,{\cal X}))$. 
Conversely, given a layout ${\cal L}= (v_1,\dots, v_n)$, we can construct a path decomposition $(P,{\cal X})_{\cal L}$ with $n$ bags, bag $i$ being $N(\{v_i, v_{i+1}, \dots, v_n\}) \cup \{v_i\}$. Clearly,
the width of this path decomposition equals $vs({\cal L})$. In particular, $\pw(G) = \min_{{\cal L} \in {\cal Q}} vs({\cal L})$~\cite{Kinnersley92} where ${\cal Q}$ is the set of all layouts.

Consider now parameter $\pw(G[N[L]]; N(L))$, for a subset $L$ of $V \setminus C$. With similar arguments, 
$$\pw(G[N[L]]; N(L)) = \min_{{\cal L}=(v_1,\dots, v_n)} \left\{\max_{n - |L| +1 \leq i \leq n} N(\{v_i, v_{i+1}, \dots, v_n\})\right\}$$
where the minimum is taken over all layouts ${\cal L}=(v_1,\dots, v_n)$ such that $L = \{v_{n-|L|+1}, \dots, v_n\}$. In full words, the minimum is taken over all layouts ending with the vertices 
of $L$.

For $L = \emptyset$, we take $\pw(G[\emptyset]; \emptyset) = 0$. For each $L$ non empty, we have (see e.g.~\cite{BodlaenderFKKT12}) that
	$$\pw(G[N[L]]; N(L)) = \max \left\{|N(L)|, \min_{u \in L} \left\{\pw(G[N[L']] ; N(L'))  \mid L' = L \setminus \{u\}\right\}\right\}.$$
Intuitively, for computing an optimal layout ending with $L$, we must try which vertex $u \in L$ occupies the position $v_{n -|L|+1}	$, then we simply take the best ordering on $L' = L \setminus \{u\}$. Using the above equation, if the values $\pw(G[N[L']] ; N(L'))$ are known for all subsets $L'$ smaller that $L$, then $\pw(G[N[L]]; N(L))$ is computable in polynomial time. Thus we can compute all values $\pw(G[N[L]]; N(L))$ by dynamic programming over all sets $L \subseteq C$ by increasing size, in time $O^*(2^{|C|})$.
\qed
\end{proof}

\begin{lemma}\label{le:glueing}
Let $G=(V,E)$ be a graph, $C$ be a vertex cover of $\overline{G}$ and $S = V \setminus C$. Then
$$\pw(G)= \min_{L \subseteq C } \max \left\{pw(G[N[L]];  N(L)]), pw(G[N[R]]; N(R)), |S \cup N(L)|-1|\right\}$$
where $R = C \setminus N[L]$.
\end{lemma}
\begin{proof}
Let $L$ be a subset of $C$ realizing the minimum of the right-hand side member of the above equation. Denote by $p$ this minimum. 
Take a path decomposition $(P_L,{\cal X}_L)$ of $G[N[L]]$, 
of width at most $p$, having $N(L)$ as highest end-bag. Also take a path decomposition 
$(P_R,{\cal X}_R)$ of $G[N[R]]$, of width at most $p$, having $N(R)$ as lowest end-bag. Observe that $N(R) \subseteq S \cup N(L)$. Therefore, by adding a bag $S \cup N(L)$ between the root $N(L)$ of $(P_L,{\cal X}_L)$ and the leaf $N(R)$ of $(P_R,{\cal X}_R)$, we obtain a path decomposition $(P, {\cal X})$ of $G$, of width $p$.

Conversely, let $(P,{\cal X})$ be an optimal rooted path decomposition of $G$ and let $p$ be its width. 
Let $X_i$ be a bag containing $S$. Let $L=L_i$ be the set of vertices 
of $G$ which only appear in bags strictly below node $i$. Observe that $X_i$ contains $S \cup N(L_i)$, in particular $|S \cup N(L_i)| -1 \leq p$. 
We construct a path decomposition of $G[N[L]]$ by deleting all nodes strictly above $i$ and by removing, from the remaining bags, all vertices that do not belong to $N[L_i]$. Observe that 
all vertices of $N(L_i)$ are still in the bag $X_i$. This
implies that $\pw(G[N[L_i]]; N(L_i)) \leq p$.

Let now $R = C \setminus N[L_i]$. We start again from the path decomposition $(P,{\cal X})$ and we construct a path decomposition of $G[N[R]]$ with $N(R)$ as a leaf bag.
Since $R = C \setminus N[L]$, we have that $N(R) \cap C \subseteq N(L)$, in particular $N(R) \subseteq X_i$. 
Remove from $(P,{\cal X})$ all bags strictly below $i$ and add below $i$ a new bag $N(R)$. From the remaining bags, remove all vertices that are not in $G[N[R]]$. We obtain a path decomposition of $G[N[R]]$ whose leaf bag is $N(R)$, thus $\pw(G[N[R]]; N(R)) \leq p$.
\qed
\end{proof}

\begin{theorem}\label{th:vcc}
The \textsc{Pathwidth} problem can be solved in time $O^*(2^{k'})$, where $k'$ is the size of the minimum vertex cover of the complement of the input graph.
\end{theorem}
\begin{proof}
We compute, using Lemma~\ref{le:wcomp}, all values $\pw(G[N[L]]; N(L))$, for each $L \subseteq C$. These values are stored in a table of size $O^*(2^{k'})$ under the label $L$. Then, for each such $L$, we let $R = C\setminus N[R]$ and we take the maximum of $\pw(G[N[L]]; N(L))$, $\pw(G[N[R]]; N(R))$ and $|S \cup N(L)| - 1$. By Lemma~\ref{le:glueing}, the minimum of all these values is the pathwidth of $G$.
\qed
\end{proof}

\end{document}